\documentclass[a4paper,final]{amsart}
\pdfoutput=1  

\usepackage[utf8]{inputenc}
\usepackage{lmodern}
\usepackage[T1]{fontenc}

\usepackage{amsmath,amssymb,amsthm,enumerate}
\usepackage[notref,notcite]{showkeys} 
\usepackage{mathpartir}

\usepackage{url}
\usepackage[hidelinks,urlcolor=black]{hyperref}

\usepackage[pass]{geometry} 

\theoremstyle{plain}
\newtheorem{theorem}{Theorem}[section]
\newtheorem*{theorem*}{Theorem}
\newtheorem{lemma}[theorem]{Lemma}
\newtheorem*{lemma*}{Lemma}
\newtheorem{corollary}[theorem]{Corollary}
\newtheorem*{corollary*}{Corollary}

\newtheorem*{proposition*}{Proposition}

\newtheorem*{claim*}{Claim}
\newtheorem*{namedtheorem}{\theoremname}

\theoremstyle{definition}
\newtheorem{definition}[theorem]{Definition}
\newtheorem*{definition*}{Definition}

\newtheorem*{notation*}{Notation}

\newtheorem*{example*}{Example}

\newtheorem*{examples*}{Examples}

\theoremstyle{remark}
\newtheorem{remark}[theorem]{Remark}
\newtheorem*{remark*}{Remark}

\newtheorem*{note*}{Note}
\newtheorem*{acknowledgments*}{Acknowledgments}

\newcommand{\IH}{IH}
\newcommand{\SIH}{SIH}
\newcommand{\der}{\vdash}
\newcommand{\co}{\colon}
\DeclareMathOperator{\dom}{dom}

\newcommand{\case}{\emph{Case}}
\newcommand{\subcase}{\emph{Subcase}}

\newcommand{\all}{\forall}

\newcommand{\NN}{\mathbb{N}}

\newcommand{\id}{\mathrm{id}}

\newcommand{\set}[1]{\{#1\}}

\newcommand{\limp}{\rightarrow}

\newcommand{\Imp}{\Rightarrow}
\newcommand{\Iff}{\Leftrightarrow}

\newcommand{\textnoti}{\text{n.i.}}
\newcommand{\noti}{\text{ n.i.}}

\renewcommand{\phi}{\varphi}

\newcommand{\II}{\mathbb{I}}
\newcommand{\FF}{\mathbb{F}}

\newcommand{\su}[2]{#1/#2}
\newcommand{\subst}[2]{(\su #1 #2)}

\newcommand{\emptyctx}{{\diamond}} 
\newcommand{\pp}{\mathsf{p}}       

\newcommand{\J}{\mathcal{J}}       

\newcommand{\N}{\mathsf{N}}
\DeclareMathOperator{\suc}{\mathsf{S}}
\newcommand{\natrec}{\mathsf{natrec}}
\DeclareMathOperator{\pred}{\mathsf{pred}}
\newcommand{\num}[1]{\underline{#1}} 

\newcommand{\Path}{\mathsf{Path}}
\newcommand{\nabs}[1]{\langle #1 \rangle}

\newcommand{\smap}{~}

\newcommand{\comp}{\mathsf{comp}}
\newcommand{\Comp}{\mathsf{fill}}

\newcommand{\Glue}{\mathsf{Glue}} 
\newcommand{\glue}{\mathsf{glue}} 
\newcommand{\unglue}{\mathsf{unglue}}

\newcommand{\Equiv}{\mathsf{Equiv}}

\newcommand{\eq}{\mathsf{equiv}}
\newcommand{\pres}{\mathsf{pres}}

\newcommand{\U}{\mathsf{U}}
\DeclareMathOperator{\ptoeq}{\mathsf{equiv}} 

\newcommand{\sone}{\mathsf{S^1}}
\newcommand{\base}{\mathsf{base}}
\DeclareMathOperator{\sloop}{\mathsf{loop}}
\newcommand{\sonerec}{\sone\textsf{-elim}}

\newcommand{\inh}[1]{\lVert#1\rVert}
\DeclareMathOperator{\inc}{\mathsf{inc}}
\newcommand{\squash}{\mathsf{squash}}
\newcommand{\hcomp}{\mathsf{hcomp}}
\newcommand{\fwd}{\mathsf{fwd}}
\newcommand{\ptElim}{\mathsf{elim}}

\newcommand{\red}{\succ}
\newcommand{\sred}{\red_\mathsf{s}}

\newcommand{\dn}{{\downarrow}}    
\newcommand{\hr}{{!}}             

\newcommand{\force}{\Vdash}
\newcommand{\lv}{\ell}           
\newcommand{\lforce}{\force_\lv} 
\newcommand{\peq}{\doteqdot}

\newcommand{\rulename}[1]{\textsc{#1}}
\newcommand{\rnc}[1]{\rulename{#1-C}}
\newcommand{\rcn}{\rnc{N}}
\newcommand{\rcpi}{\rnc{Pi}}
\newcommand{\rcsi}{\rnc{Si}}
\newcommand{\rcpa}{\rnc{Pa}}
\newcommand{\rcgl}{\rnc{Gl}}
\newcommand{\rcu}{\rnc{U}}
\newcommand{\rcni}{\rnc{Ni}}
\newcommand{\rcpt}{\rnc{Pt}}

\newcommand{\rne}[1]{\rulename{#1-E}}
\newcommand{\ren}{\rne{N}}
\newcommand{\repi}{\rne{Pi}}
\newcommand{\resi}{\rne{Si}}
\newcommand{\repa}{\rne{Pa}}
\newcommand{\regl}{\rne{Gl}}
\newcommand{\reu}{\rne{U}}
\newcommand{\reni}{\rne{Ni}}
\newcommand{\rept}{\rne{Pt}}

\newcommand{\nir}[1]{\inferrule* [right=#1]}

\title{Canonicity for Cubical Type Theory}
\date{\today}

\author{Simon Huber}
\address{Department of Computer Science and Engineering, University of
  Gothenburg, SE-412 96 G\"oteborg, Sweden}
\email{simon.huber@cse.gu.se}

\begin{document}

\begin{abstract}
  Cubical type theory is an extension of Martin-L\"of type theory
  recently proposed by Cohen, Coquand, M\"ortberg, and the author
  which allows for direct manipulation of $n$-dimensional cubes and
  where Voevodsky's Univalence Axiom is provable.  In this paper we
  prove canonicity for cubical type theory: any natural number in a
  context build from only name variables is judgmentally equal to a
  numeral.  To achieve this we formulate a typed and deterministic
  operational semantics and employ a computability argument adapted to
  a presheaf-like setting.
\end{abstract}

\maketitle

\section{Introduction}
\label{sec:introduction}

Cubical type theory as presented in~\cite{CohenCoquandHuberMortberg15}
is a dependent type theory which allows one to directly argue about
$n$-dimensional cubes, and in which function extensionality and
Voevodsky's Univalence Axiom~\cite{Voevodsky10} are provable.  Cubical
type theory is inspired by a \emph{constructive} model of dependent
type theory in cubical sets~\cite{CohenCoquandHuberMortberg15} and a
previous variation thereof~\cite{BezemCoquandHuber14,Huber15}.  One of
its important ingredients is that expressions can depend on
\emph{names} to be thought of as ranging over a formal unit
interval~$\II$.

Even though the consistency of the calculus already follows from its
model in cubical sets, desired---and expected---properties like
normalization and decidability of type checking are not yet
established.  This note presents a first step in this direction by
proving \emph{canonicity} for natural numbers in the following form:
given a context $I$ of the form $i_1 : \II, \dots, i_k : \II$, $k \ge
0$, and a derivation of $I \der u : \N$, there is a unique $n \in \NN$
with $I \der u = \suc^n 0 : \N$.  This $n$ can moreover be effectively
calculated.  Canonicity in this form also gives an alternative proof
of the consistency of cubical type theory (see
Corollary~\ref{cor:consistency}).

The main idea to prove canonicity is as follows.  First, we devise an
operational semantics given by a typed and deterministic weak-head
reduction included in the judgmental equality of cubical type theory.
This is given for general contexts although we later on will only use
it on terms whose only free variables are name variables, i.e.,
variables of type~$\II$.  One result we obtain is that our reduction
relation is ``complete'' in the sense that any term in a name context
whose type is the natural numbers can be reduced to one in weak-head
normal form (so to zero or a successor).  Second, we will follow
Tait's computability method~\cite{Tait67,MartinLoef98} and devise
computability predicates on \emph{typed} expressions in name contexts
and corresponding computability relations (to interpret judgmental
equality).  These computability predicates are indexed by the list of
free name variables of the involved expressions and should be such
that substitution induces a cubical set structure on them.  This poses
a major difficulty given that the reduction relation is in general not
closed under name substitutions.  A solution is to require for
computability that reduction should behave ``coherently'' with
substitution: simplified, reducing an expression and then substituting
should be related, by the computability relation, to first
substituting and then reducing.  A similar condition appeared
independently in the Computational Higher Type Theory of Angiuli,
Harper, and Wilson~\cite{AngiuliHarperWilson16,AngiuliHarperWilson17}
and Angiuli and Harper~\cite{AngiuliHarper16} who work in an untyped
setting; they achieve similar results but for a theory not
encompassing the Univalence Axiom.

In a way, our technique can be considered as a presheaf extension of
the computability argument given
in~\cite{AbelScherer16,AbelCoquandMannaa16}; the latter being an
adaption of the former using a typed reduction relation instead.  A
similar extension of this technique has been used to show the
independence of Markov's principle in type theory
\cite{CoquandMannaa16}.

The rest of the paper is organized as follows.  In
Section~\ref{sec:reduction} we introduce the typed reduction relation.
Section~\ref{sec:computability} defines the computability predicates
and relations and shows their important properties.  In
Section~\ref{sec:soundness} we show that cubical type theory is sound
w.r.t.\ the computability predicates; this entails canonicity.
Section~\ref{sec:hits} sketches how to adapt the computability
argument for the system extended with the circle and propositional
truncation, and we deduce an existence property for existentials
defined as truncated $\Sigma$-types.  We conclude by summarizing and
listing further work in the last section.  We assume that the reader
is familiar with cubical type theory as given
in~\cite{CohenCoquandHuberMortberg15}.

\section{Reduction}
\label{sec:reduction}

In this section we give an operational semantics for cubical type
theory in the form of a typed and deterministic weak-head reduction.
Below we will introduce the relations $\Gamma \der A \red B$ and
$\Gamma \der u \red v : A$.  These relations are deterministic in the
following sense: if $\Gamma \der A \red B$ and $\Gamma \der A \red C$,
then $B$ and $C$ are equal as expressions (i.e., up to
$\alpha$-equivalence); and, if $\Gamma \der u \red v : A$ and $\Gamma
\der u \red w : B$, then $v$ and $w$ are equal as expressions.
Moreover, these relations entail judgmental equality, i.e., if $\Gamma
\der A \red B$, then $\Gamma \der A = B$, and if $\Gamma \der u \red v
: A$, then $\Gamma \der u = v : A$.

For a context $\Gamma \der$, a \emph{$\Gamma$-introduced} expression is an
expression whose outer form is an introduction, so one of the form
\begin{gather*}
  0, \suc u, \N, \lambda x : A. u, (x : A) \limp B, (u,v), ( x :
  A) \times B, \U, \nabs i u, \Path\,A\,u\,v,\\
  [ \phi_1 \smap t_1, \dots, \phi_n \smap t_n], \glue\,[\phi \mapsto
  t]\,a, \Glue\,[\phi \mapsto (T,w)]\,A,
\end{gather*}
where we require $\phi \neq 1 \mod \Gamma$ (which we from now on write
as $\Gamma \der \phi \neq 1 : \FF$) for the latter two cases, and in
the case of a system (third to last) we require $\Gamma \der \phi_1
\lor \dots \lor \phi_n = 1 : \FF$ but $\Gamma \der \phi_k \neq 1 :
\FF$ for each $k$.  In case $\Gamma$ only contains object and interval
variable declarations (and \emph{no} restrictions $\Delta,\psi$) we
simply refer to $\Gamma$-introduced as \emph{introduced}.  In such a
context, $\Gamma \der \phi = \psi : \FF$ if{f} $\phi = \psi$ as
elements of the face lattice $\FF$; since $\FF$ satisfies the
disjunction property, i.e.,
\[
\phi \lor \psi = 1 \Imp \phi = 1 \text{ or } \psi = 1,
\]
a system as above will never be introduced in such a context without
restrictions.  We call an expression \emph{non-introduced} if it is
not introduced and abbreviate this as ``\textnoti'' (often this is
referred to as neutral or non-canonical).  A $\Gamma$-introduced
expression is normal w.r.t.\ $\Gamma \der \cdot \red \cdot$ and
$\Gamma \der \cdot \red \cdot : A$.
\\

We will now give the definition of the reduction relation starting
with the rules concerning basic type theory.
\begin{mathparpagebreakable}
  \inferrule {\Gamma \der u \red v : A \\
    \Gamma \der A = B}
  { \Gamma \der u \red v : B}
  \\
  \inferrule { %
    \Gamma, x : \N \der C \\
    \Gamma \der z : C \subst x 0 \\
    \Gamma \der s : (x : \N) \limp C \limp C \subst x {\suc x} \\
  } %
  { \Gamma \der \natrec \, 0 \, z \, s \red z : C \subst x 0} \and %
  \inferrule { %
    \Gamma \der t : \N \\
    \Gamma, x : \N \der C \\
    \Gamma \der z : C \subst x 0\\
    \Gamma \der s : (x : \N) \limp C \limp C \subst x {\suc x} \\
  } %
  { \Gamma \der \natrec \, (\suc t) \, z \, s \red s \, t \, (\natrec
    \, t \, z \, s ) : C \subst x {\suc t} } %
  \and %
  \inferrule { %
    \Gamma \der t \red t' : \N \\
    \Gamma, x : \N \der C \\
    \Gamma \der z : C \subst x 0 \\
    \Gamma \der s : (x : \N) \limp C \limp C \subst x {\suc x} \\
  } %
  { \Gamma \der \natrec \, t \, z \, s\red \natrec \, t' \, z \, s: C
    \subst x {t'} } %
  \and %
  \inferrule {\Gamma, x : A \der t : B \\
    \Gamma \der u : A} %
  {\Gamma \der (\lambda x : A. t) \, u \red t \subst x u : B \subst x
    u} \and %
  \inferrule {\Gamma \der t \red t' : (x : A) \to B \\
    \Gamma \der u : A} %
  {\Gamma \der t \, u \red t' u : B \subst x u} %
  \and %
  \inferrule { %
    \Gamma, x : A \der B\\
    \Gamma \der u : A \\
    \Gamma \der v : B \subst x u} %
  { \Gamma \der (u,v).1 \red u : A} \and %
  \inferrule { %
    \Gamma \der t \red t' : (x : A) \times B } %
  { \Gamma \der t.1 \red t'.1 : A} \and %
  \inferrule { %
    \Gamma, x : A \der B\\
    \Gamma \der u : A \\
    \Gamma \der v : B \subst x u} %
  { \Gamma \der (u,v).2 \red v : B \subst x u} \and %
  \inferrule { %
    \Gamma \der t \red t' : (x : A) \times B } %
  { \Gamma \der t.2 \red t'.2 : B \subst x {t'.1}}
\end{mathparpagebreakable}
Note, $\natrec\,t \, z \, s$ is not considered as an application
(opposed to the presentation in \cite{CohenCoquandHuberMortberg15});
also the order of the arguments is different to have the main premise
as first argument.

Next, we give the reduction rules for $\Path$-types.  Note, that like
for $\Pi$-types, there is no $\eta$-reduction or expansion, and also
there is no reduction for the end-points of a path.
\begin{mathparpagebreakable}
  \inferrule { %
    \Gamma \der A \\
    \Gamma, i : \II \der t : A\\
    \Gamma \der r : \II}
  { \Gamma \der (\nabs i t) \, r \red t \subst i r : A}
  \and %
  \inferrule {\Gamma \der t \red t' : \Path \, A \, u\, v \\
    \Gamma \der r : \II} %
  {\Gamma \der t \, r \red t' r : A } %
\end{mathparpagebreakable}
The next rules concern reductions for $\Glue$.
\begin{mathparpagebreakable}
  \inferrule {\Gamma \der A \\
    \Gamma, \phi \der T \\
    \Gamma,\phi \der w : \Equiv \, T \, A \\
    \Gamma \der \phi = 1 : \FF} %
  {\Gamma \der \Glue \, [\phi \mapsto (T,w)] \, A \red T} %
  \and %
  \inferrule { %
    \Gamma,\phi \der w : \Equiv \, T \, A \\
    \Gamma,\phi \der t : T \\
    \Gamma \der a : A [\phi \mapsto w.1 \, t]\\
    \Gamma \der \phi = 1 : \FF} %
  {\Gamma \der \glue \, [\phi \mapsto t] \, a \red t : T}
  \and %
  \inferrule { %
    \Gamma,\phi \der w : \Equiv \, T \, A \\
    \Gamma,\phi \der t : T\\
    \Gamma \der a : A [\phi \mapsto w.1 \, t]\\
    \Gamma \der \phi \neq 1 : \FF
  } %
  {\Gamma \der \unglue\,[\phi \mapsto w]\, (\glue \, [\phi \mapsto t]
    \, a) \red a : A}
  \and %
  \inferrule { %
    \Gamma,\phi \der w : \Equiv \, T \, A \\
    \Gamma \der u : \Glue \, [\phi \mapsto (T,w)] \, A\\
    \Gamma \der \phi = 1 : \FF} %
  {\Gamma \der \unglue\,[\phi \mapsto w]\,u \red w.1 \, u : A}
  \and %
  \inferrule { %
    \Gamma \der u \red u' : \Glue \, [\phi \mapsto (T,w)] \, A \\
    \Gamma \der \phi \neq 1 : \FF
  } %
  { \Gamma \der \unglue\,[\phi \mapsto w]\, u \red \unglue\,[\phi
    \mapsto w]\, u' : A }
\end{mathparpagebreakable}
Note that in~\cite{CohenCoquandHuberMortberg15} the annotation $[\phi
\mapsto w]$ of $\unglue$ was left implicit.  The rules for systems are
given by:
\begin{mathparpagebreakable}
  \inferrule { %
    \Gamma \der \phi_1 \lor \dots \lor \phi_n = 1 : \FF\\
    \Gamma, \phi_i \der A_i ~ (1 \le i \le n)\\
    \Gamma, \phi_i \land \phi_j \der A_i = A_j ~ (1\le i,j\le n)\\
    \text{$k$ minimal with }\Gamma \der \phi_k = 1 : \FF } %
  { \Gamma \der [ \phi_1 \smap A_1, \dots, \phi_n \smap A_n] \red
    A_k} %
  \and %
  \inferrule { %
    \Gamma \der \phi_1 \lor \dots \lor \phi_n = 1 : \FF\\
    \Gamma \der A\\
    \Gamma, \phi_i \der t_i : A ~ (1 \le i \le n)\\
    \Gamma, \phi_i \land \phi_j \der t_i = t_j : A ~ (1\le i,j\le n)\\
    \text{$k$ minimal with }\Gamma \der \phi_k = 1 : \FF } %
  { \Gamma \der [ \phi_1 \smap t_1, \dots, \phi_n \smap t_n] \red t_k
    : A}
\end{mathparpagebreakable}
The reduction rules for the universe are:
\begin{mathparpagebreakable}
  \inferrule { %
    \Gamma \der A \red B : \U } %
  { \Gamma \der A \red B} %
  \and %
  \inferrule {\Gamma \der A : \U \\
    \Gamma, \phi \der T : \U \\
    \Gamma,\phi \der w : \Equiv \, T \, A \\
    \Gamma \der \phi = 1 : \FF} %
  {\Gamma \der \Glue \, [\phi \mapsto (T,w)] \, A \red T : \U} %
\end{mathparpagebreakable}
Finally, the reduction rules for compositions are given as follows.
\begin{mathparpagebreakable}
  \inferrule { %
    \Gamma, i : \II \der A \red B \\
    \Gamma \der \phi : \FF\\
    \Gamma,\phi, i : \II \der u : A \\
    \Gamma \der u_0 : A (i0) [\phi \mapsto u (i0)]\\
  } { \Gamma \der \comp^i \, A \, [\phi \mapsto u] \, u_0 \red %
    \comp^i \, B \, [\phi \mapsto u] \, u_0 : B(i1) %
  } %
  \and %
  \inferrule { %
    \Gamma \der \phi : \FF\\
    \Gamma,\phi,i:\II \der u : \N \\
    \Gamma,\phi,i:\II \der u = 0 : \N } %
  { \Gamma \der \comp^i \, \N \, [ \phi \mapsto u] \, 0 \red 0 : \N} %
  \and %
  \inferrule { %
    \Gamma \der \phi : \FF\\
    \Gamma,\phi,i:\II \der u : \N\\
    \Gamma,\phi,i:\II \der w : \N\\
    \Gamma,\phi,i:\II \der u = \suc w : \N\\
    \Gamma \der u_0 : \N\\
    \Gamma,\phi \der u (i0) = \suc u_0 : \N\\
  } %
  { \Gamma \der \comp^i \, \N \, [ \phi \mapsto u] \, (\suc u_0)
    \red %
    \suc (\comp^i \, \N \, [ \phi \mapsto \pred u] \, u_0) : \N} %
\end{mathparpagebreakable}
Here $\pred$ is the usual predecessor function defined using
$\natrec$.\footnote{This trick allows us that we never have to reduce
  in the system of a composition when defining composition for natural
  numbers, which also gives that reduction over $\Gamma$ never refers
  to reduction in a restricted context $\Gamma,\phi$ (given that
  $\Gamma$ is not restricted).  If we would instead directly require
  $u$ above to be of the form $\suc u'$, we would have to explain
  reductions for systems like $[ (i=0) \smap (\suc t), (i=1) \smap
  (\suc {t'})]$ and more generally how reduction and systems
  interact.}
\begin{mathparpagebreakable}
  \inferrule { %
    \Gamma \der \phi : \FF\\
    \Gamma, \phi, i : \II \der u : \N \\
    \Gamma \der u_0 : \N [\phi \mapsto u(i0)] \\
    \Gamma \der u_0 \red v_0 : \N } %
  {\Gamma \der \comp^i \, \N \, [ \phi \mapsto u] \, u_0 \red %
    \comp^i \, \N \, [ \phi \mapsto u] \, v_0 : \N} \and %
  \inferrule { %
    \Gamma \der \phi : \FF\\
    \Gamma, i : \II \der A \\
    \Gamma, i : \II, x : A \der B \\
    \Gamma, \phi, i : \II \der u : (x : A) \to B \\
    \Gamma \der u_0 : ((x : A) \to B) (i0) [\phi \mapsto u(i0)] \\
  } %
  { \Gamma \der \comp^i \, ((x : A) \to B) \, [\phi \mapsto u] \, u_0
    \red
    \\
    \lambda y : A(i1). \comp^i \, B \subst{x}{\bar y} \, [\phi \mapsto
    u \, \bar y] \, (u_0 \, \bar y (i0)): (x : A(i1)) \to B(i1)
    \\
    \text{where }y' = \Comp^i\,A \subst{i} {1-i}\,[]\,y \text { and
    }\bar y = y' \subst{i}{1-i} } %
  \and %
  \inferrule { %
    \Gamma \der \phi : \FF\\
    \Gamma, i : \II \der A \\
    \Gamma, i : \II, x:A \der B \\
    \Gamma, \phi, i : \II \der u : (x : A) \times B \\
    \Gamma \der u_0 : ((x : A) \times B) (i0) [\phi \mapsto u(i0)] \\
  } %
  { \Gamma \der \comp^i \, ((x : A) \times B) \, [\phi \mapsto u] \,
    u_0 \red \\
    (v (i1), \comp^i \, B \subst {x} {v} \, [\phi \mapsto u.2]
    \,(u_0.2)) : (x : A (i1)) \times B(i1)
    \\
    \text{where } v = \Comp^i \, A \, [\phi \mapsto u.1] \, (u_0 .1)}
  \and %
  \inferrule { %
    \Gamma \der \phi : \FF\\
    \Gamma, i : \II \der A\\
    \Gamma, i : \II \der v : A\\
    \Gamma, i : \II \der w : A\\
    \Gamma, \phi, i : \II \der u : \Path\,A\,v\,w \\
    \Gamma \der u_0 : \Path\,A(i0)\,v(i0)\,w(i0) [\phi \mapsto u(i0)] \\
  } %
  { \Gamma \der \comp^i \, (\Path\,A\,v\,w) \, [\phi \mapsto u] \, u_0
    \red \\
    \nabs j \, \comp^i \, A \,[ (j=0) \mapsto v, (j=1) \mapsto w, \phi
    \mapsto u \, j] \, (u_0 \, j) : \Path\,A(i1)\,v(i1)\,w(i1)} %
  \and
  \inferrule { %
    \Gamma, i : \II \der A \\
    \Gamma, i : \II \der \phi : \FF\\
    \Gamma, i : \II \der \phi \neq 1 : \FF\\
    \Gamma, i : \II, \phi \der T \\
    \Gamma, i : \II, \phi \der w : \Equiv \, T \, A \\
    \Gamma \der \psi : \FF\\
    \Gamma,\psi, i : \II \der u : \Glue \, [\phi \mapsto (T,w)] \, A\\
    \Gamma \der u_0 : (\Glue \, [\phi \mapsto (T,w)] \,
    A) (i0) [\psi \mapsto u (i0)]\\
  } %
  { \Gamma \der \comp^i \, (\Glue \, [\phi \mapsto (T,w)] \, A) \, %
    [ \psi \mapsto u] \, u_0 \red\\
    \glue \, [\phi (i1) \mapsto t_1] \, a_1 : %
    (\Glue \, [\phi \mapsto (T,w)] \, A)(i1)
  } %
\end{mathparpagebreakable}
Here $a_1$ and $t_1$ are defined like
in~\cite{CohenCoquandHuberMortberg15}, i.e., given by
\begin{align*}
  a &= \unglue\,[\phi \mapsto w]\, u
  && \Gamma,i : \II ,\psi
  \\
  a_0 &= \unglue\,[\phi(i0) \mapsto w(i0)]\, u_0
  && \Gamma
  \\
  \delta &= \forall i.\varphi && \Gamma
  \\
  a_1' &= \comp^i\,A\,[\psi\mapsto a]\,a_0 && \Gamma
  \\
  t_1' &= \comp^i\,T\,[\psi\mapsto u]\,u_0 && \Gamma,\delta
  \\
  \omega &=\pres^i\,w\,[\psi\mapsto u]\,u_0 && \Gamma,\delta
  \\
  (t_1,\alpha) &= \eq\,w(i1)\,[\delta \mapsto (t'_1,\omega),
  \psi \mapsto (u(i1),\nabs {j}{a_1'})]\,a_1' && \Gamma,\varphi(i1)
  \\
  a_1 &= \comp^j\,A(i1)\,[\varphi(i1)\mapsto \alpha\,j,\psi\mapsto
  a(i1)]\,a_1' && \Gamma
\end{align*}
where we indicated the intended context on the right.

\begin{mathparpagebreakable}
  \inferrule { %
    \Gamma \der \phi : \FF\\
    \Gamma, \phi, i : \II \der u : \U \\
    \Gamma \der u_0 : \U [\phi \mapsto u(i0)] \\
  } %
  { \Gamma \der \comp^i \, \U \, [\phi \mapsto u] \, u_0 \red
    \Glue\,[\phi \mapsto (u(i1), \ptoeq^i {u \subst {i} {1-i}})]\,u_0 :
    \U
  }
\end{mathparpagebreakable}
Here $\ptoeq^i$ is defined as in~\cite{CohenCoquandHuberMortberg15}.
This concludes the definition of the reduction relation.
\\

For $\Gamma \der A$ we write $A \hr_\Gamma$ if there is $B$ such that
$\Gamma \der A \red B$; in this case $B$ is uniquely determined by $A$
and we denote $B$ by $A \dn_\Gamma$; if $A$ is normal we set $A
\dn_\Gamma$ to be $A$.  Similarly for $\Gamma \der u : A$, $u
\hr_\Gamma^A$ and $u \dn^A_\Gamma$.  Note that if a term or type has a
reduct it is non-introduced.  We usually drop the subscripts and
sometimes also superscripts since they can be inferred.

From now on we will mainly consider contexts $I,J,K,\dots$ only built
from dimension name declarations; so such a context is of the form
$i_1 : \II, \dots, i_n : \II$ for $n \ge 0$.  We sometimes write $I,i$
for $I, i : \II$.  Substitutions between such contexts will be denoted
by $f,g,h,\dots$.  The resulting category with such name contexts $I$
as objects and substitutions $f \co J \to I$ is reminiscent of the
category of cubes as defined in
\cite[Section~8.1]{CohenCoquandHuberMortberg15} with the difference
that the names in a contexts $I$ are ordered and not sets.  This
difference is not crucial for the definition of computability
predicates in the next section but it simplifies notations.  (Note
that if $I'$ is a permutation of $I$, then the substitution assigning
to each name in $I$ itself is an isomorphism $I' \to I$.)  We write $r
\in \II (I)$ if $I \der r : \II$, and $\phi \in \FF (I)$ if $I \der
\phi : \FF$.

Note that in general reductions $I \der A \red B$ or $I \der u \red v
: A$ are not closed under substitutions $f \co J \to I$.  For example,
if $u$ is a system $[ (i = 0) \smap u_1, 1 \smap u_2]$, then $i \der u
\red u_2 : A$ (assuming everything is well typed), but $\der u (i0)
\red u_1 (i0) : A(i0)$ and $u_1,u_2$ might be chosen that $u_1(i0)$
and $u_2(i0)$ are judgmentally equal but not syntactically (and even
normal by considering two $\lambda$-abstractions where the body is not
syntactically but judgmentally equal).  Another example is when $u$ is
$\unglue\,[\phi \mapsto w]\, (\glue\,[\phi \mapsto t]\,a)$ with $\phi
\neq 1$ and with $f \co J \to I$ such that $\phi f = 1$; then $u$
reduces to $a$, but $uf$ reduces to $wf.1\,(\glue\,[\phi f \mapsto t
f]\,a f)$ which is in general not syntactically equal to $af$.

We write $I \der A \sred B$ and $I \der u \sred v : A$ if the
respective reduction is closed under name substitutions.  That is, $I
\der A \sred B$ whenever $J \der A f \red B f$ for all $f \co J \to
I$.  Note that in the above definition, all the rules which do not
have a premise with a \emph{negated} equation in $\FF$ and which do
not have a premise referring to another reduction are closed under
substitution.

\section{Computability Predicates}
\label{sec:computability}

In this section we define computability predicates and establish the
properties we need for the proof of Soundness in the next section.  We
will define when a type is \emph{computable} or \emph{forced}, written
$I \lforce A$, when two types are forced equal, $I \lforce A = B$,
when an element is computable or forced, $I \lforce u : A$, and when
two elements are forced equal, $I \lforce u = v : A$.  Here $\lv$ is
the \emph{level} which is either $0$ or $1$, the former indicating
smallness.

The definition is given as follows: by main recursion on $\lv$ (that
is, we define ``$\force_0$'' before ``$\force_1$'') we define by
induction-recursion~\cite{Dybjer00}
\begin{align*}
  I &\lforce A
  \\
  I &\lforce A = B
  \\
  I &\lforce u : A && \text{by recursion on } I \lforce A
  \\
  I &\lforce u = v : A && \text{by recursion on } I \lforce A
\end{align*}
where the former two are mutually defined by induction, and the latter
two mutually by recursion on the derivation of $I \lforce A$.
Formally, $I \lforce A$ and $I \lforce A = B$ are witnessed by
derivations for which we don't introduce notations since the
definitions of $I \lforce u : A$ and $I \lforce u = v : A$ don't
depend on the derivation of $I \lforce A$.  Each such derivation has a
height as an ordinal, and often we will employ induction not only on
the structure of such a derivation but on its height.

Note that the arguments and definitions can be adapted to a hierarchy
of universes by allowing $\lv$ to range over a (strict) well-founded
poset.

%
We write $I \lforce A \peq B$ for the conjunction of $I \lforce A$, $I
\lforce B$, and $I \lforce A = B$.
For $\phi \in \FF (I)$ we write $f \co J \to I,\phi$ for $f \co J \to
I$ with $\phi f = 1$; furthermore we write
\begin{align*}
  I,\phi &\lforce A &&\text{for} &&\all f \co J \to I,\phi
  \, (J \lforce A f) \And I,\phi \der A,\\
  I,\phi &\lforce A = B &&\text{for} &&\all f \co J \to I,\phi
  \, (J \lforce A f = B f) \And I,\phi \der A = B,\\
  I,\phi &\lforce u : A &&\text{for} &&\all f \co J \to I,\phi \,
  (J \lforce u f : A f) \And I,\phi \der u : A,\\
  I,\phi &\lforce u = v: A &&\text{for} &&\all f \co J \to I,\phi \, (J
  \lforce u f = v f : A f) \And I,\phi \der u = v : A.
\end{align*}
where the last two abbreviations need suitable premises to make sense.
Note that $I,1 \lforce A$ is a priori stronger than $I \lforce A$;
that these notions are equivalent follows from the Monotonicity Lemma
below. Moreover, the definition is such that $I \der \J$ whenever $I
\lforce \J$ (where $\J$ is any judgment form); it is shown in
Remark~\ref{rem:restricted-contexts} that the condition $I,\phi \der
\J$ in the definition of $I,\phi \lforce \J$ is actually not needed
and follows from the other.
\\

\noindent\fbox{$I \lforce A$} assuming $I \der A$ (i.e., the
rules below all have a suppressed premise $I \der A$).
\begin{mathparpagebreakable}
  \nir{\rcn} { {} } {I \lforce \N} %
  \and
  \nir
  {\rcpi} { %
    I,1 \lforce A \\ %
    I, x : A \der B\\
    \all f \co J \to I \all u (J \lforce u : A f \Imp J \lforce B (f,
    \su x {u})) \\ %
    \all f \co J \to I \all u,v (J \lforce u = v : A f \Imp J \lforce B
    (f,\su x {u}) \peq B (f, \su x {v})) %
  } %
  {I \lforce (x : A) \to B} %
  \and
  \nir {\rcsi} { %
    I,1 \lforce A \\ %
    I, x : A \der B\\
    \all f \co J \to I \all u (J \lforce u : A f \Imp J \lforce B (f,
    \su x {u})) \\ %
    \all f \co J \to I \all u,v (J \lforce u = v : A f \Imp J \lforce
    B (f, \su x {u}) \peq B(f, \su x {v})) %
  } %
  {I \lforce (x : A) \times B} %
  \and
  \nir {\rcpa} {I,1 \lforce A \\ I \lforce a_0 : A \\ I \lforce a_1 : A}
  {I \lforce \Path \, A \, a_0 \, a_1} %
  \and
  \nir {\rcgl} { %
    1 \neq \phi \in \FF (I) \\
    I,1 \lforce A \\
    I, \phi \lforce \Equiv \, T\, A \\
    I, \phi \lforce w : \Equiv \, T\, A \\
    I, \phi \lforce \Glue \, [ \phi \mapsto (T, w)] \, A } %
  { I \lforce \Glue \, [ \phi \mapsto (T, w)] \, A}
  \and %
  \nir {\rcu} { %
  } %
  { I \force_1 \U}
  \and %
  \nir {\rcni} { %
    A \noti \\
    \all f \co J \to I (A f \hr \And J \lforce A f \dn) \\\\
    \all f \co J \to I\all g \co K \to J %
    (K \lforce A f \dn g = A f g\dn) } %
  {I \lforce A}
\end{mathparpagebreakable}
Note, that the rule $\rcgl$ above is not circular, as for any $f \co J
\to I,\phi$ we have $\phi f = 1$ and so $(\Glue \, [ \phi \mapsto
(T, w)] \, A)f$ is non-introduced.
\medskip

\noindent\fbox{$I \lforce A = B$} assuming $I \lforce A$, $I \lforce
B$, and $I \der A = B$ (i.e., each rule below has the suppressed
premises $I \lforce A$, $I \lforce B$, and $I \der A = B$).
\begin{mathparpagebreakable}
  \nir{\ren} { {} } {I \lforce \N = \N} %
  \and
  \nir {\repi} {I,1 \lforce A = A' \\ %
    I, x : A \der B = B'\\
    \all f \co J \to I \all u (J \lforce u : A f \Imp J \lforce B (f,
    \su x {u}) = B' (f, \su x {u})) %
  } %
  {I \lforce (x : A) \to B = (x : A') \to B'} %
  \and
  \nir
  {\resi}  {I,1 \lforce A = A' \\ %
    I, x : A \der B = B'\\
    \all f \co J \to I \all u (J \lforce u : A f \Imp J \lforce B (f,
    \su x {u}) = B' (f, \su x {u})) %
  } %
  {I \lforce (x : A) \times B = (x : A') \times B'} %
  \and
  \nir {\repa} {I,1 \lforce A = B \\ I \lforce a_0 = b_0 : A \\ I \lforce
    a_1 = b_1 : A}
  {I \lforce \Path \, A \, a_0 \, a_1 = \Path \, B \, b_0 \, b_1} %
  \and
  \nir {\regl} { %
    1 \neq \phi \in \FF (I) \\
    I,1 \lforce A = A' \\
    I, \phi \lforce \Equiv \, T\, A = \Equiv \, T' \, A'\\
    I, \phi \lforce w = w' : \Equiv \, T\, A \\
    I, \phi \lforce \Glue \, [ \phi \mapsto (T, w)] \, A = \Glue \, [
    \phi \mapsto (T', w')] \, A' } %
  { I \lforce \Glue \, [ \phi \mapsto (T, w)] \, A = \Glue \, [ \phi
    \mapsto (T', w')] \, A'} \and
  \nir {\reu} { %
  } %
  { I \force_1 \U = \U}
  \and %
  \nir {\reni} {A \text{ or } B \noti \\ \all f \co J \to I (J \lforce
    A f \dn = B f \dn)} %
  {I \lforce A = B}
\end{mathparpagebreakable}
\medskip

\noindent\fbox{$I \lforce u : A$ by induction on $I \lforce A$} assuming
$I \der u : A$.  We distinguish cases on the derivation of $I \lforce
A$.

\case\ \rcn.
\begin{mathpar}
  \inferrule { } {I \lforce 0 : \N} \and %
  \inferrule {I \lforce u : \N} {I \lforce \suc u : \N} \and %
  \inferrule { %
    u \noti \\
    \all f \co J \to I (u f \hr^\N \And J \lforce u f \dn^\N : \N) \\\\
    \all f\co J \to I \all g \co K \to J (K \lforce u f \dn^\N g = u
    fg \dn^\N: \N) } %
  { I \lforce u : \N}
\end{mathpar}

\case\ \rcpi.
\begin{mathpar}
  \inferrule { %
    \all f \co J \to I \all u (J \lforce u : A f \Imp J \lforce w f\,
    u : B (f,\su x u ))
    \\
    \all f \co J \to I \all u,v (J \lforce u = v : Af \Imp J \lforce w
    f \,u = w f \, v : B (f, \su x u )) } %
  { I \lforce w : (x : A) \to B}
\end{mathpar}

\case\ \rcsi.
\begin{mathpar}
  \inferrule {I \lforce u.1 : A \\ I \lforce u.2 : B \subst x {u.1}}%
  {I \lforce u : (x : A) \times B}
\end{mathpar}

\case\ \rcpa.
\begin{mathpar}
  \inferrule { %
    \all f \co J \to I \all r \in \II (J) (J \lforce u f \, r : A f) \\
    I \lforce u \, 0 = a_0 : A \\
    I \lforce u \, 1 = a_1 : A } %
  { I \lforce u : \Path \, A \, a_0 \, a_1}
\end{mathpar}

\case\ \rcgl.
\begin{mathpar}
  \inferrule { %
    I,\phi \lforce u : \Glue \, [ \phi \mapsto (T,w)]  \, A \\
    \all f \co J \to I \all w' (J,\phi f \force w' = wf :
    \Equiv\,Tf\,Af \Imp \qquad\qquad \\\\
    \qquad J \force \unglue\,[\phi f \mapsto w']\,uf =
    \unglue\,[\phi f \mapsto wf]\,uf : A f)
  }
  { I \lforce u : \Glue \, [ \phi \mapsto (T,w)] \, A}
\end{mathpar}
Later we will see that from the premises of \rcgl\ we get $I \force w
= w : \Equiv\,T\,A$, and the second premise above implies in
particular $I \force \unglue\,[\phi \mapsto w]\,u : A$; the
quantification over other possible equivalences is there to ensure
invariance for the annotation.

\case\ \rcu.
\begin{mathpar}
  \inferrule { %
    I \force_0 A
  } %
  { I \force_1 A : \U}
\end{mathpar}

\case\ \rcni.
\begin{mathpar}
  \inferrule { %
    \all f \co J \to I (J \lforce u f : A f \dn) } %
  { I \lforce u : A}
\end{mathpar}
\medskip

\noindent\fbox{$I \lforce u = v : A$ by induction on $I \lforce A$}
assuming $I \lforce u : A$, $I \lforce v : A$, and $I \der u = v :A$.
(I.e., each of the rules below has the suppressed premises $I \lforce
u : A$, $I \lforce v : A$, and $I \der u = v : A$, but they are not
arguments to the definition of the predicate.  This is subtle since
in, e.g., the rule for pairs we only know $I \lforce v.2 : B \subst x
{v.1}$ not $I \lforce v.2 : B \subst x {u.1}$.)  We distinguish cases
on the derivation of $I \lforce A$.

\case\ \rcn.
\begin{mathpar}
  \inferrule { } {I \lforce 0 = 0 : \N} \and %
  \inferrule {I \lforce u = v : \N} {I \lforce \suc u = \suc v : \N}
  \and %
  \inferrule { u \text { or } v \noti \\ 
    \all f  (J \lforce u f \dn^\N = v f \dn^\N : \N)} %
  { I \lforce u = v : \N}
\end{mathpar}

\case\ \rcpi.
\begin{mathpar}
  \inferrule { %
    \all f \co J \to I \all u (J \lforce u : A f \Imp J \lforce w f \,
    u = w' f \, u
    : B (f, \su x u )) \\
  } %
  { I \lforce w = w' : (x : A) \to B}
\end{mathpar}

\case\ \rcsi.
\begin{mathpar}
  \inferrule {I \lforce u.1 = v.1 : A \\ I \lforce u.2 = v.2 : B
    \subst x {u.1}} {I \lforce u = v : (x : A) \times B}
\end{mathpar}

\case\ \rcpa.
\begin{mathpar}
  \inferrule { %
    \all f \co J \to I \all r \in \II(J) (J \lforce u f \, r = v f \,
    r : A f) } %
  { I \lforce u = v : \Path \, A \, a_0 \, a_1}
\end{mathpar}

\case\ \rcgl.
\begin{mathpar}
  \inferrule { %
    I,\phi \lforce u = v : \Glue \, [ \phi \mapsto (T,w)] \, A\\
    I,1 \lforce \unglue\,[\phi \mapsto w]\,u = \unglue\,[\phi \mapsto
    w]\, v : A %
  } %
  { I \lforce u = v : \Glue \, [ \phi \mapsto (T,w)] \, A}
\end{mathpar}

\case\ \rcu.
\begin{mathpar}
  \inferrule { %
    I \force_0 A = B
  } %
  { I \force_1 A = B : \U}
\end{mathpar}

\case\ \rcni.
\begin{mathpar}
  \inferrule { %
    \all f \co J \to I (J \lforce u f = v f: A f \dn) } %
  { I \lforce u = v : A}
\end{mathpar}


Note that the definition is such that $I \lforce A = B$ implies $I
\lforce A$ and $I \lforce B$; and, likewise, $I \lforce u = v : A$
gives $I \lforce u : A$ and $I \lforce v :A$.


\begin{remark}
  \label{rem:computability-def}
  \leavevmode
  \begin{enumerate}
  \item In the rule \reni\ and the rule for $I \lforce u = v : \N$ in
    case $u$ or $v$ are non-introduced we suppressed the premise that
    the reference to ``$\dn$'' is actually well defined; it is easily
    seen that if $I \lforce A$, then $A \dn$ is well defined, and
    similarly for $I \lforce u : \N$, $u \dn^\N$ is well defined.
  \item It follows from the substitution lemma below that $I \lforce A$
    whenever $A$ is non-introduced and $I \der A \sred B$ with $I
    \lforce B$.  (Cf.\ also the Expansion Lemma below.)
  \item Note that once we also have proven transitivity, symmetry, and
    monotonicity, the last premise of \rcni\ in the definition of $I
    \lforce A$ (and similarly in the rule for non-introduced naturals)
    can be restated as $J \lforce A f \dn = A \dn f$ for all $f \co J
    \to I$.
  \end{enumerate}
\end{remark}

\begin{lemma}
  \label{lem:indep-of-derivation}
  The computability predicates are independent of the derivation,
  i.e., if we have two derivations trees $d_1$ and $d_2$ of $I \lforce
  A$, then
  \begin{align*}
    I \lforce^{d_1} u : A &\Iff I \lforce^{d_2} u : A, \text{ and}\\
    I \lforce^{d_1} u = v : A &\Iff I \lforce^{d_2} u = v : A
  \end{align*}
  where $\lforce^{d_i}$ refers to the predicate induced by $d_i$.
\end{lemma}
\begin{proof}
  By main induction on $\lv$ and a side induction on the derivations
  $d_1$ and $d_2$.  Since the definition of $I \lforce A$ is syntax
  directed both $d_1$ and $d_2$ are derived by the same rule.  The
  claim thus follows from the \IH.
\end{proof}

\begin{lemma}
  \label{lem:cpred-complete}
  \leavevmode
  \begin{enumerate}
  \item If $I \lforce A$, then $I \der A$ and:
    \begin{enumerate}
    \item $I \lforce u : A \Imp I \der u : A$,
    \item $I \lforce u = v : A \Imp I \der u = v : A$.
    \end{enumerate}
  \item If $I \lforce A = B$, then $I \der A = B$.
  \end{enumerate}
\end{lemma}

\begin{lemma}
  \label{lem:lvl-cumulative}
  \leavevmode
  \begin{enumerate}
  \item If $I \force_0 A$, then:
    \begin{enumerate}
    \item $I \force_1 A$
    \item $I \force_0 u : A \Iff I \force_1 u : A$
    \item $I \force_0 u = v : A \Iff I \force_1 u = v : A$
    \end{enumerate}
  \item If $I \force_0 A = B$, then $I \force_1 A = B$.
  \end{enumerate}
\end{lemma}
\begin{proof}
  By simultaneous induction on $I \force_0 A$ and $I \force_0 A = B$.
\end{proof}

We will write $I \force A$ if there is a derivation of $I \lforce A$
for some $\lv$; etc.  Such derivations will be ordered
lexicographically, i.e., $I \force_0 A$ derivations are ordered before
$I \force_1 A$ derivations.

\begin{lemma}
  \label{lem:per-field}
  \leavevmode
  \begin{enumerate}
  \item\label{item:typ-refl} $I \lforce A \Imp I \lforce A = A$
  \item\label{item:ter-refl} $I \lforce A \And I \lforce u : A \Imp I
    \lforce u = u : A$
  \end{enumerate}
\end{lemma}
\begin{proof}
  Simultaneously, by induction on $\lv$ and side induction on $I
  \lforce A$.  In the case \rcgl, to see~\eqref{item:ter-refl}, note
  that from the assumption $I \force u : B$ with $B$ being
  $\Glue\,[\phi\mapsto (T,w)]\,A$ we get in particular
  \[
  I,\phi \force w = w : \Equiv\,T\,A \Imp I \force \unglue\,[\phi
  \mapsto w]\,u = \unglue\,[\phi \mapsto w]\,u : A.
  \]
  But by \IH, the premise follows from $I,\phi \force w :
  \Equiv\,T\,A$; moreover, $I,\phi \force u = u : B$ is immediate by
  \IH, showing $I \force u = u : B$.
\end{proof}

\begin{lemma}[Monotonicity/Substitution]
  \label{lem:subst}
  For $f \co J \to I$ we have
  \begin{enumerate}
  \item\label{item:subst-typ} $I \lforce A \Imp J \lforce A f$,
  \item $I \lforce A = B \Imp J \lforce A f = B f$,
  \item $I \lforce A \And I \lforce u : A \Imp J \lforce u f : A f$,
  \item $I \lforce A \And I \lforce u = v : A \Imp J \lforce u f = v f :
    A f$.
  \end{enumerate}
  Moreover, the respective heights of the derivations don't
  increase.
\end{lemma}
\begin{proof}
  By induction on $\lv$ and side induction on $I \lforce A$ and $I
  \lforce A=B$.  The definition of computability predicates and
  relations is lead such that this proof is immediate.  For instance,
  note for~\eqref{item:subst-typ} in the case $\rcgl$, i.e.,
  \begin{equation*}
    \nir {\rcgl} { %
      1 \neq \phi \in \FF (I) \\
      I,1 \lforce A \\
      I, \phi \lforce \Equiv \, T\, A \\
      I, \phi \lforce w : \Equiv \, T\, A \\
      I, \phi \lforce \Glue \, [ \phi \mapsto (T, w)] \, A } %
    { I \lforce \Glue \, [ \phi \mapsto (T, w)] \, A}
  \end{equation*}
  we distinguish cases: if $\phi f = 1$, then $J \lforce \Glue \, [
  \phi f \mapsto (T f, w f)] \, A f$ by the premise $ I, \phi \lforce
  \Glue \, [ \phi \mapsto (T, w)] \, A$; in case $\phi f \neq 1$ we
  can use the same rule again.
\end{proof}

\begin{lemma}
  \label{lem:red-closure}
  \leavevmode
  \begin{enumerate}
  \item\label{item:7} $I \force A \Imp I \force A \dn$
  \item\label{item:12} $I \force A = B\Imp I \force A \dn = B \dn$
  \item\label{item:13} $I \force A \And I \force u : A \Imp I
    \force u : A \dn$
  \item\label{item:14} $I \force A \And I \force u = v : A \Imp I
    \force u = v: A \dn$
  \item\label{item:11} $I \force u : \N \Imp I \force u \dn : \N$
  \item\label{item:16} $I \force u = v : \N \Imp I \force u \dn = v
    \dn: \N$
  \end{enumerate}
  Moreover, the respective heights of the derivations don't increase.
\end{lemma}
\begin{proof}
  \eqref{item:7} By induction on $I \force A$. All
  cases were $A$ is an introduction are immediate since then $A \dn$
  is $A$.  It only remains the case \rcni:
  \[
  \nir {\rcni} { %
    A \noti \\
    \all f \co J \to I (A f \hr \And J \force A f \dn) \\
    \all f \co J \to I\all g \co K \to J (K \force A f \dn g = A f
    g\dn) } %
  {I \force A}
  \]
  We have $I \force A \dn$ as this is one of the premises.

  \eqref{item:11}  By induction on $I
  \force u : \N$ similarly to the last paragraph.

  \eqref{item:12} By induction on $I \force A = B$.  The only case
  where a reduct may happen is \reni, in which $I \force A \dn = B
  \dn$ is a premise.  Similar for \eqref{item:16}.

  \eqref{item:13} and \eqref{item:14}: By induction on $I \force A$,
  where the only interesting case is \rcni, in which what we have to
  show holds by definition.
\end{proof}

\begin{lemma}
  \label{lem:eq-resp-trans}
  \leavevmode
  \begin{enumerate}
  \item\label{item:eq-resp}
    If $I \force A = B$, then
    \begin{enumerate}
    \item\label{item:eq1} $I \force u : A \Iff I \force u : B$, and
    \item\label{item:eq2} $I \force u = v : A \Iff I \force u = v : B$.
    \end{enumerate}
  \item\label{item:typ-trans} $I \force A = B \And I \force B = C \Imp
    I \force A = C$
  \item\label{item:ter-trans}  Given $I \force A$ we get
    \[
    I \force u = v : A \And I \force v = w : A \Imp I \force u = w :
    A.
    \]
  \item\label{item:typ-sym} $I \force A = B \Imp I \force B = A$
  \item\label{item:ter-sym} $ I \force A \And I \force u = v : A \Imp
    I \force v = u : A$
  \end{enumerate}
\end{lemma}
\begin{proof}
  We prove the statement for ``$\lforce$'' instead of ``$\force$'' by
  main induction on $\lv$ (i.e., we prove the statement for
  ``$\force_0$'' before the statement for ``$\force_1$''); the
  statement for ``$\force$'' follows then from
  Lemma~\ref{lem:lvl-cumulative}.

  Simultaneously by threefold induction on $I \lforce A$, $I \lforce
  B$, and $I \lforce C$. (Alternatively by induction on the (natural)
  sum of the heights of $I \lforce A$, $I \lforce B$, and $I \lforce
  C$; we only need to be able to apply the \IH\ if the complexity of
  at least one derivation decreases and the others won't increase.)
  In the proof below we will omit $\lv$ to simplify notation, except
  in cases where the level matters.

  \eqref{item:eq-resp} By distinguishing cases on $I \force A = B$.
  We only give the argument for \eqref{item:eq1} as \eqref{item:eq2}
  is very similar except in case \regl.  The cases \ren\ and \reu\ are
  trivial.

  \case\ \repi.  Let $I \force w : (x :A) \to B$ and we show $I \force
  w : (x :A') \to B'$.  For $f \co J \to I$ let $J \force u : A' f$;
  then by \IH\ (since $J \force A f = A' f$) we get $J \force u : A
  f$, and thus $J \force w f \, u : B (f, \su x u )$; again by \IH\ we
  obtain $J \force w f \,u : B' (f, \su x u )$.  Now assume $J \force
  u = v: A' f$; so by \IH, $J \force u = v: A f$, and thus $J \force w
  f \, u = w f \, v : B (f, \su x u )$.  Again by \IH, we conclude $J
  \force w f \, u = w f \, v : B' (f, \su x u )$.  Thus we have proved $I
  \force w : (x :A') \to B'$.

  \case\ \resi.  Let $I \force w : (x :A) \times B$ and we show
  $I \force w : (x :A') \times B'$.  We have $I \force w.1 : A$ and
  $I \force w.2 : B \subst {x} {w.1}$.  So by \IH,
  $I \force w.1 : A'$; moreover, we have $I \force B \subst {x} {w.1}
  = B' \subst {x} {w.1}$; so, again by \IH, we conclude with $I \force
  w.2 : B' \subst {x} {w.1}$.

  \case\ \repa.  Let $I \force u : \Path\,A\,a_0\,a_1$ and we show $I
  \force u : \Path\,B\,b_0\,b_1$.  Given $f \co J \to I$ and $r \in
  \FF (J)$ we have $J \force u f \, r : A f$ and thus $J \force u f \,
  r : B f$ by \IH.  We have to check that the endpoints match: $I
  \force u \, 0 = a_0 : A$ by assumption; moreover, $I \force a_0 =
  b_0 : A$, so by \IH~\eqref{item:ter-trans}, $I \force u \, 0 = b_0 :
  A$, thus again using the \IH, $I \force u \, 0 = b_0 : B$.

  \case\ \regl.  Abbreviate $\Glue \, [ \phi \mapsto (T, w)] \, A$ by
  $D$, and $\Glue \, [ \phi \mapsto (T', w')] \, A'$ by~$D'$.

  \eqref{item:eq1} Let $I \force u : D$, i.e., $I,\phi \force u : D$ and
  \begin{equation}
    \label{eq:38}
    J \force \unglue\,[\phi f \mapsto w'']\,uf = %
    \unglue\,[\phi f \mapsto wf]\,uf : A f
  \end{equation}
  whenever $f \co J \to I$ and $J,\phi f \force w'' = wf :
  \Equiv\,Tf\,Af$.  Directly by \IH\ we obtain $I,\phi \force u : D'$.
  Now let $f \co J \to I$ and $J,\phi f \force w''=w'f :
  \Equiv\,T'f\,A'f$; by \IH, also $J,\phi f \force w''=w'f :
  \Equiv\,Tf\,Af$.  Moreover, we have $J,\phi f \force w f = w' f :
  \Equiv\,Tf\,Af$, hence \eqref{eq:38} gives (together with symmetry
  and transitivity, applicable by \IH)
  \begin{align*}
    J &\force \unglue\,[\phi f \mapsto w'']\,uf = %
    \unglue\,[\phi f \mapsto wf]\,uf : A f, \text{ and}\\
    J &\force \unglue\,[\phi f \mapsto w' f]\,uf = %
    \unglue\,[\phi f \mapsto wf]\,uf : A f.
  \end{align*}
  Hence, transitivity and symmetry (which we can apply by \IH) give
  that the above left-hand sides are forced equal of type $A f$,
  applying the \IH~\eqref{item:eq2} gives that they are forced equal
  of type $A'f$, and thus $I \force u : D'$.

  \eqref{item:eq2} Let $I \force u = v : D$, so we have $I,\phi \force
  u = v : D$ and
  \begin{equation}
    \label{eq:39}
    I \force \unglue\,[\phi \mapsto w]\, u = \unglue\,[\phi \mapsto
    w]\, v : A
  \end{equation}
  By \IH, we get $I,\phi \force u = v : D'$ from $I,\phi \force u = v
  : D$.  Note that we also have $I \force u : D$ and $I \force v : D$,
  and thus
  \begin{align*}
    I &\force \unglue\,[\phi \mapsto w]\, u = \unglue\,[\phi \mapsto
    w']\, u : A, \text{ and}\\
    I &\force \unglue\,[\phi \mapsto w]\, v = \unglue\,[\phi \mapsto
    w']\, v : A
  \end{align*}
  and thus with \eqref{eq:39} and transitivity and symmetry (which we
  can apply by \IH) we obtain $I \force \unglue\,[\phi \mapsto w']\, u
  = \unglue\,[\phi \mapsto w']\, v : A$, hence also at type $A'$ by
  \IH.  Therefore we proved $I \force u = v : D'$.

  \case\ \reni.  Let $I \force u : A$; we have to show $I \force u :
  B$.

  \subcase\ $B$ is non-introduced.  Then we have to show $J \force u f
  : B f \dn$ for $f \co J \to I$.  We have $J \force A f \dn = B f
  \dn$ and since $I \force B$ is non-introduced, the derivation $J
  \force B f \dn$ is shorter than $I \force B$, and the derivation $J
  \force A f \dn$ is not higher than $I \force A$ by
  Lemma~\ref{lem:red-closure}.  Moreover, also $J \force uf : Af$ so
  by Lemma~\ref{lem:red-closure}~\eqref{item:13} we get $J \force u f
  : A f \dn$, and hence by \IH, $J \force u f : B f \dn$.

  \subcase\ $B$ is introduced.  We have $I \force u : A \dn$ and $I
  \force A \dn = B \dn$ but $B \dn$ is $B$, and $I \force A \dn$ has a
  shorter derivation than $I \force A$, so $I \force u : B$ by \IH.

  \eqref{item:typ-trans} Let us first handle the cases where $A$, $B$,
  or $C$ is non-introduced.  It is enough to show $J \force A f \dn =
  C f \dn$ (if $A$ and $C$ are both introduced, this entails $I \force
  A = C$ for $f$ the identity).  We have $J \force A f \dn = B f \dn$
  and $J \force B f \dn = C f \dn$.  None of the respective
  derivations get higher (by Lemma~\ref{lem:red-closure}) but one gets
  shorter since one of the types is non-introduced.  Thus the claim
  follows by \IH.

  It remains to look at the cases where all are introduced; in this
  case both equalities have to be derived by the same rule.  We
  distinguish cases on the rule.

  \case\ \ren. Trivial. \case\ \resi.  Similar to \repi\ below.
  \case\ \repa\ and \regl.  Use the \IH.  \case\ \reu. Trivial.

  \case\ \repi. Let use write $A$ as $(x : A') \to A''$ and similar
  for $B$ and $C$.  We have $I \force A' = B'$ and $I \force B' = C'$,
  and so by \IH, we get $I\force A' = C'$; for $J \force u : A' f$
  where $f \co J \to I$ it remains to be shown that $J \force A''
  (f,\su x u ) = C'' (f, \su x u )$.  By \IH, we also have $J \force u
  : B' f$, so we have
  \[
  J \force A'' (f,\su x u ) = B'' (f,\su x u ) \text{ and } %
  J \force B'' (f,\su x u ) = C'' (f,\su x u )
  \]
  and can conclude by the \IH.

  \eqref{item:ter-trans} By cases on $I \force A$.  All cases follow
  immediately using the \IH, except for \rcn\ and \rcu.  In case \rcn,
  we show transitivity by a side induction on the (natural) sum of the
  height of the derivations $I \force u = v : \N$ and $I \force v = w
  : \N$.  If one of $u$,$v$, or $w$ is non-introduced, we get that one
  of the derivations $J \force u f \dn = v f \dn : \N$ and $J \force v
  f \dn = w f \dn : \N$ is shorter (and the other doesn't get higher),
  so by \SIH, $J \force u f \dn = w f \dn : \N$ which entails $I
  \force u = w : \N$.  Otherwise, $I \force u = v : \N$ and $I \force
  v = w : \N$ have to be derived with the same rule and $I \force u =
  w : \N$ easily follows (using the \SIH\ in the successor case).

  In case \rcu, we have $I \force_1 u = v : \U$ and $I \force_1 v = w
  : \U$, i.e., $I \force_0 u = v$ and $I \force_0 v = w$. We want to
  show $I \force_1 u = w : \U$, i.e., $I \force_0 u = w$.  But by \IH
  ($\lv$), we can already assume the lemma is proven for $\lv = 0$,
  hence can use transitivity and deduce $I \force_0 u = w$.
  %
  %

  The proofs of \eqref{item:typ-sym} and \eqref{item:ter-sym} are by
  distinguishing cases and are straightforward.
\end{proof}

\begin{remark}
  \label{rem:transitivity-and-pi}
  Now that we have established transitivity, proving computability for
  $\Pi$-types can also be achieved as follows.  Given we have $I
  \force (x :A) \to B$ and derivations $I \der w : (x :A) \to B$, $I
  \der w' : (x :A) \to B$, and $I \der w = w' : (x :A) \to B$, then $I
  \force w = w' : (x :A) \to B$ whenever we have
  \begin{equation*}
    \all f \co J \to I \all u, v (J \force u = v : Af \Imp J \force w
    f \, u = w' f \, v : B (f, \su x u )).
  \end{equation*}
  (In particular, this gives $I \force w : (x :A) \to B$ and $I \force
  w' : (x :A) \to B$.)

  Likewise, given $I \der (x :A) \to B$, $I \der (x :A') \to B'$, and
  $I \der (x :A) \to B = (x :A') \to B'$, we get $I \force (x :A) \to
  B = (x :A') \to B'$ whenever $I \force A = A'$ and
  \begin{equation*}
    \all f \co J \to I \all u, v (J \force u = v : Af \Imp %
    J \force B (f, \su x u ) = B' (f, \su x v )).
  \end{equation*}
\end{remark}

\pagebreak
\begin{lemma}
  \label{lem:red-eq}
  \leavevmode
  \begin{enumerate}
  \item\label{item:8} $I \force A \Imp I \force A = A \dn$
  \item\label{item:nat-red-cl} $I \force u : \N \Imp I \force u = u
    \dn : \N$
  \end{enumerate}
\end{lemma}
\begin{proof}
  \eqref{item:7} We already proved $I \force A \dn$ in
  Lemma~\ref{lem:red-closure}~\eqref{item:7}.  By induction on $I
  \force A$.  All cases where $A$ is an introduction are immediate
  since then $A \dn$ is $A$.  It only remains the case \rcni:
  \[
  \nir {\rcni} { %
    A \noti \\
    \all f \co J \to I (A f \hr \And J \force A f \dn) \\
    \all f \co J \to I\all g \co K \to J (K \force A f \dn g = A f
    g\dn) } %
  {I \force A}
  \]
  We now show $I \force A = A \dn$; since $A$ is non-introduced we
  have to show $J \force A f \dn = (A \dn f) \dn$ for $f \co J \to I$.
  $I \force A \dn$ has a shorter derivation than $I \force A$, thus so
  has $J \force A \dn f$; hence by \IH, $J \force A \dn f = (A \dn f)
  \dn$.  We also have $J \force A \dn f = A f \dn$ by definition of $I
  \force A$, and thus we obtain $J \force A f \dn = (A \dn f) \dn$
  using symmetry and transitivity.

  \eqref{item:nat-red-cl} Similar, by induction on $I \force u : \N$.
\end{proof}

\begin{lemma}[Expansion Lemma]
  \label{lem:exp}
  Let $I \lforce A$ and $I \der u : A$; then:
  \begin{mathpar}
    \inferrule{ %
      \all f \co J \to I (uf \hr^{A f} \And J \lforce uf \dn^{A f} : A
      f) \\
      \all f \co J \to I %
      (J \lforce u f \dn = u \dn f : A f)
    } %
    { I \lforce u : A \And I \lforce u = u \dn : A}
  \end{mathpar}
  In particular, if $I \der u \sred v : A$ and $I \lforce v : A$, then
  $I \lforce u :A$ and $I \lforce u = v : A$.
\end{lemma}
\begin{proof}
  By induction on $I \force A$.  We will omit the level annotation
  $\lv$ whenever it is inessential.

  \case\ \rcn.  We have to show $K \force u f \dn g = u f g \dn : \N$
  for $f \co J \to I$ and $g \co K \to J$; we have $J \force u f \dn =
  u \dn f : \N$, thus $K \force u f \dn g = u \dn f g : \N$. Moreover,
  $K \force u \dn f g = u f g \dn : \N$ by assumption, and thus by
  transitivity $K \force u f \dn g = u \dn f g = u f g \dn : \N$.
  (Likewise one shows that the data in the premise of the lemma is
  closed under substitution.)

  $I \force u = u \dn : \N$ holds by
  Lemma~\ref{lem:red-closure}~\eqref{item:nat-red-cl}.

  \case\ \rcpi.  First, let $J \force a : A f$ for $f \co J \to I$.
  We have
  \[
  K \der (uf \, a) g \red (u {fg}) \dn \, (a g) : B (fg, \su x {a g})
  \]
  for $g \co K \to J$, and also $K \force (u {fg}) \dn \, (a g) : B
  (fg, \su x {a g})$ and we have the compatibility condition
  \begin{align*}
    K \force (u f \, a) g \dn %
    &= ((ufg) \dn) \, (a g) %
    = (uf \dn g) \, (a g) \\%
    &= (uf \dn \, a) g %
    = (uf \, a) \dn g %
    : B (fg, \su x {a g}),
  \end{align*}
  so by \IH, $J \force u f \, a : B (f, \su x a )$ and $J \force u f
  \, a = u f \dn \, a : B (f,\su x a )$.  Since also $J \force u f \dn
  = u \dn f : ((x : A) \to B)f$ we also get $J \force u f \, a = u \dn
  f \, a : B (f,\su x a )$.

  Now if $J \force a = b : A f$, we also have $J \force a : A f$ and
  $J \force b : A f$, so like above we get $J \force u f \, a = uf\dn
  \, a : B (f,\su x a )$ and $J \force u f \, b = uf \dn \, b : B
  (f,\su x b )$ (and thus also $J \force u f \, b = uf\dn \, b : B
  (f,\su x a )$).  Moreover, $J \force uf \dn \, a = u f \dn \, b : B
  (f,\su x a )$ and hence we can conclude $J \force u f \, a = u f \,
  b : B (f,\su x a )$ by transitivity and symmetry.  Thus we showed
  both $I \force u : (x : A) \to B$ and $I \force u = u\dn: (x : A)
  \to B$.

  \case\ \rcsi. Clearly we have $(u.1 f) \dn = (u f \dn).1$, $J \force
  (u f \dn).1 : A f$, and
  \[
  J \force (u.1 f) \dn = (u f \dn).1 = (u \dn f) .1 = (u \dn .1) f =
  (u.1) \dn f : A f
  \]
  so the \IH\ gives $I \force u.1 : A$ and $I \force u.1 = (u \dn).1 :
  A$.  Likewise $(u.2 f) \dn = (u f \dn).2$ and $J \force (u f \dn).2
  : B (f, \su x {uf\dn.1})$, hence also $J \force (u f \dn).2 : B (f,
  \su x {uf.1})$; as above one shows $J \force (u.2 f) \dn = u.2 \dn f
  : B (f, \su x {uf.1})$, applying the \IH\ once more to obtain $I
  \force u.2 = u\dn.2 : B\subst x {u.1})$ which was what remained to
  be proven.

  \case\ \rcpa. Let us write $\Path\,A\,v\,w$ for the type and let $f
  \co J \to I$, $r \in \II (J)$, and $g \co K \to J$.  We have
  \[
  K \der (u f \, r) g \red (ufg) \dn \, (r g) : A f g
  \]
  and $K \force (uf g) \dn \, (rg) : A fg$; moreover,
  \[
  K \force (u f \, r) g \dn %
  = (ufg) \dn \, (r g) %
  = (uf \dn g) \, (r g) %
  = (uf \dn \, r)  g %
  = (uf \, r) \dn  g %
  : A fg.
  \]
  Thus by \IH, $J \force u f\, r : A f$ and
  \begin{equation}
    \label{eq:30}
    J \force uf \, r = uf\dn
    \,r = u \dn f \, r : A f.
  \end{equation}
  So we obtain $I \force u \, 0 = u\dn \, 0 = v : A$ and $I \force u
  \, 1 = u\dn \, 1 = w : A$, and hence $I \force u : \Path\,A\,v\,w$;
  $I \force u = u \dn : \Path\,A\,v\,w$ follows from~\eqref{eq:30}.

  \case\ \rcgl. Abbreviate $\Glue \, [\phi \mapsto (T,w)] \, A$ by
  $B$.  Note that we have $\phi \neq 1$.  First, we claim that for any
  $f \co J \to I$, $J \force b : Bf$, and $J,\phi f \force w' = wf :
  \Equiv\,Tf\,Af$,
  \begin{equation}
    \label{eq:8}
    J,\phi f \force \unglue\,[\phi f\mapsto w']\,b = w'.1 \, b : A f.
  \end{equation}
  (In particular both sides are computable.) Indeed, for $g \co K \to
  J$ with $\phi f g = 1$ we have that
  \[
  K \der
  (\unglue\,[\phi f\mapsto w']\,b) g \sred w'g.1 \, (b g) : A f g
  \]
  and $K \force w'g.1 \, (b g) : A f g$ since $I,\phi \force B = T$
  (which follows from Lemma~\ref{lem:red-eq}~\eqref{item:8}).  Thus by
  \IH\ ($J \force A f$ has a shorter derivation than $I \force B$), $K
  \force (\unglue\,[\phi f\mapsto w']\,b) g = (w'.1 \, b)g : A f g$ as
  claimed.

  Next, let $f \co J \to I$ such that $\phi f = 1$; then using the
  \IH\ ($J \force Bf$ has a shorter derivation than $I \force B$), we
  get $J \force uf : Bf$ and $J \force uf = u f \dn : Bf$, and hence
  also $J \force u f = u \dn f : B f$ (since $J \force u f \dn = u \dn
  f : B f$).  That is, we proved
  \begin{equation}
    \label{eq:7}
    I, \phi \force u : B \text{ and } I, \phi \force u = u \dn : B.
  \end{equation}

  We will now first show
  \begin{equation}
    \label{eq:24}
    J \force \unglue\,[\phi f \mapsto w']\,uf = (\unglue\,[\phi f
    \mapsto w']\,uf) \dn : A f
  \end{equation}
  for and $f \co J \to I$ and $J,\phi f \force w' = wf :
  \Equiv\,Tf\,Af$.  We can assume that w.l.o.g.\ $\phi f \ne 1$, since
  if $\phi f = 1$, $J \force u f : B f$ by \eqref{eq:7}, and
  \eqref{eq:24} follows from \eqref{eq:8} noting that its right-hand
  side is the reduct.  We will use the \IH\ to show \eqref{eq:24}, so
  let us analyze the reduct:
  \begin{equation}
    \label{eq:10}
    (\unglue\,[\phi f \mapsto w']\,uf) g \dn =
    \begin{cases}
      (w'g.1) \, (u fg) & \text{if } \phi f g = 1,\\
      \unglue\,[\phi f \mapsto w' g]\, (ufg\dn) & \text{otherwise.}
    \end{cases}
  \end{equation}
  where $g \co K \to J$.  In either case, the reduct is computable: in
  the first case, use \eqref{eq:7} and $J \force w'.1 : T \to A$
  together with the observation $I,\phi \force B = T$; in the second
  case this follows from $J \force ufg \dn : B fg$.  In order to apply
  the \IH, it remains to verify
  \begin{equation*}
    K \force (\unglue\,[\phi f \mapsto w']\,uf) g \dn =
    (\unglue\,[\phi f \mapsto w']\,uf) \dn g : A f g.
  \end{equation*}
  In case $\phi f g \ne 1$, we have
  \begin{align*}
    K \force{} &\unglue\,[\phi fg \mapsto w'g]\,(ufg \dn) %
    \\
    &= \unglue\,[\phi fg \mapsto w f g]\,(ufg \dn) %
    && \text{since }K \force ufg \dn : Bfg\\
    &= \unglue\,[\phi fg \mapsto w f g]\,(uf \dn g) %
    && \text{since }K \force ufg \dn = uf \dn g : Bfg\\
    &= \unglue\,[\phi fg \mapsto w' g]\,(uf \dn g) : A fg %
    && \text{since }K \force uf \dn g : Bfg
  \end{align*}
  which is what we had to show in this case.  In case $\phi f g = 1$,
  we have to prove
  \begin{equation}
    \label{eq:25}
    K \force (w'g.1) \, (u fg) =
    \unglue\,[\phi f \mapsto w'g]\,(uf \dn g) : A f g.
  \end{equation}
  But by \eqref{eq:7} we have $K \force u f g = u f g \dn = u f \dn g
  : B f g$, so also $K \force (w'g.1) \, (u fg) = (w'g.1) \, (u f\dn
  g) : A f g$, so \eqref{eq:25} follows from \eqref{eq:8} using $J
  \force u f \dn : Bf$.  This concludes the proof of \eqref{eq:24}.

  As $w'$ could have been $wf$ we also get
  \begin{equation}
    \label{eq:32}
    J \force \unglue\,[\phi f \mapsto wf]\,uf = (\unglue\,[\phi f
    \mapsto wf]\,uf) \dn : A f.
  \end{equation}

  In order to prove $I \force u : B$ it remains to check that the
  left-hand side of \eqref{eq:24} is forced equal to the left-hand
  side of \eqref{eq:32}; so we can simply check this for the
  respective right-hand sides: in case $\phi f = 1$, these are $w'.1
  \, uf$ and $wf.1 \, uf$, respectively, and hence forced equal since
  $J \force w' = wf : \Equiv\,Tf\,Af$; in case $\phi f \ne 1$, we have
  to show
  \[
  J \force \unglue\, [\phi f \mapsto w']\,(u f \dn) = \unglue\, [\phi
  f \mapsto w f]\,(u f \dn) : A f
  \]
  which simply follows since $J \force uf \dn : B f$.

  In order to prove $I \force u = u \dn : B$ it remains to check
  \[
  I \force \unglue\,[\phi  \mapsto w]\,u = \unglue\,[\phi  \mapsto
  w]\,(u \dn ) : A,
  \]
  but this is \eqref{eq:32} in the special case where $f$ is the
  identity.

  \case\ \rcu.  Let us write $B$ for $u$.  We have to prove $I
  \force_1 B : \U$ and $I \force_1 B = B \dn : \U$, i.e., $I \force_0
  B$ and $I \force_0 B = B \dn$.  By
  Lemma~\ref{lem:red-eq}~\eqref{item:8}, it suffices to prove the
  former.  For $f \co J \to I$ we have
  \[
  J \der B f \red B f \dn^\U : \U
  \]
  and hence also
  \[
  J \der B f \red B f \dn^\U
  \]
  i.e., $B f \hr$, and $B f \dn$ is $B f \dn^U$; since $J \force_1 Bf
  \dn : \U$ we have $J \force_0 B f \dn$, and likewise $J \force_0 B f
  \dn = B \dn f$.  Moreover, if also $g \co K \to J$, we obtain $K
  \force_0 B f g \dn = B \dn f g$ from the assumption.  Hence $K
  \force_0 B f \dn g = B \dn f g = B f g \dn$, therefore $I \force_0
  B$ what we had to show.

  \case\ \rcni. Then $I \force A \dn$ has a shorter derivation than $I
  \force A$; moreover, for $f \co J \to I$ we have $J \der u f \red uf
  \dn^{Af} : A f$ so also $J \der u f \red uf \dn^{A f} : A \dn f$
  since $J \der A f = A \dn f$.  By
  Lemma~\ref{lem:red-closure}~\eqref{item:8}, $I \force A = A \dn$ so
  also $J \force uf \dn : A \dn f$ and $J \force uf \dn = u \dn f : A
  \dn f$, and hence by \IH, $I \force u : A \dn$ and $I \force u = u
  \dn : A \dn$, so also $I \force u : A$ and $I \force u = u \dn : A$
  using $I \force A = A \dn$ again.
\end{proof}

\section{Soundness}
\label{sec:soundness}

The aim of this section is to prove canonicity as stated in the
introduction.  We will do so by showing that each computable instance
of a judgment derived in cubical type theory is computable (allowing
free name variables)---this is the content of the Soundness Theorem
below.

We first extend the computability predicates to contexts and
substitutions.

\noindent\fbox{${} \force \Gamma$} assuming $\Gamma \der {}$.
\begin{mathpar}
  \inferrule { } { {} \force \emptyctx } %
  \and %
  \inferrule {{} \force \Gamma \\ i \notin \dom(\Gamma) } %
  { {} \force \Gamma, i : \II} %
  \and %
  \inferrule { %
    {} \force \Gamma \\
    \Gamma \der \phi : \FF
  } %
  { {} \force \Gamma, \phi } %
  \and %
  \inferrule { %
    {} \force \Gamma \\
    \all I \all \sigma ( I \force \sigma : \Gamma \Imp I \force A
    \sigma) \\
    \all I \all \sigma, \tau ( I \force \sigma = \tau :
    \Gamma \Imp I \force A \sigma = A \tau) \\
    x \notin \dom (\Gamma)%
  } %
  { {} \force \Gamma,x:A } %
\end{mathpar}

\noindent\fbox{$I \force \sigma : \Gamma$} by induction on ${} \force
\Gamma$ assuming $I \der \sigma : \Gamma$.
\begin{mathpar}
  \inferrule { } { I \force () : \emptyctx} %
  \and %
  \inferrule {I \force \sigma : \Gamma \\ r \in \II (I)} %
  { I \force (\sigma, \su i r ) : \Gamma, i : \II}
  \and %
  \inferrule { %
    I \force \sigma : \Gamma \\
    \phi \sigma  = 1 } %
  { I \force \sigma : \Gamma,\phi} %
  \and %
  \inferrule {I \force \sigma : \Gamma \\
    I \force u  : A \sigma} %
  { I \force (\sigma, \su x u ) : \Gamma,x:A} %
\end{mathpar}

\noindent\fbox{$I \force \sigma = \tau : \Gamma$} by induction on ${}
\force \Gamma$, assuming $I \force \sigma : \Gamma$, $I \force \tau :
\Gamma$, and $I \der \sigma = \tau : \Gamma$.
\begin{mathpar}
  \inferrule { } { I \force () = () : \emptyctx} %
  \and %
  \inferrule {I \force \sigma = \tau : \Gamma \\ r \in \II (I)} %
  { I \force (\sigma, \su i r ) = (\tau, \su i r ) : \Gamma, i : \II}
  \and %
  \inferrule { %
    I \force \sigma = \tau : \Gamma \\
    \phi \sigma = \phi \tau = 1 } %
  { I \force \sigma = \tau : \Gamma,\phi} %
  \and %
  \inferrule {I \force \sigma = \tau : \Gamma \\
    I \force u = v : A \sigma} %
  { I \force (\sigma, \su x u ) = (\tau, \su x v ) : \Gamma,x:A} %
\end{mathpar}

We write $I \force r : \II$ for $r \in \II (I)$, $I \force r = s :
\II$ for $r = s \in \II (I)$, and likewise $I \force \phi : \FF$ for
$\phi \in \FF (I)$, $I \force \phi = \psi : \FF$ for $\phi = \psi \in
\FF (I)$.  In the next definition we allow $A$ to be $\FF$ or $\II$,
and also correspondingly for $a$ and $b$ to range over interval and
face lattice elements.
\begin{definition}\label{def:snd}
  \begin{align*}
    &\Gamma \models {} &:\Iff &&& {} \force \Gamma
    \\
    &\Gamma \models A = B &:\Iff&&& %
    \Gamma \der A = B \And \Gamma \models \And\\ &&&&& %
    \all I, \sigma, \tau (I \force \sigma = \tau : \Gamma \Imp I
    \force A \sigma = B \tau)
    \\
    & \Gamma \models A &:\Iff &&& \Gamma \der A \And \Gamma \models A
    = A
    \\
    & \Gamma \models a = b : A &:\Iff &&& %
    \Gamma \der a = b : A \And \Gamma \models A \And
    \\ &&&&& %
    \all I, \sigma, \tau (I \force \sigma = \tau : \Gamma \Imp I
    \force a \sigma = b \tau : A \sigma)
    \\
    & \Gamma \models a : A &:\Iff &&& \Gamma \der a : A \And \Gamma
    \models a = a : A
    \\
    & \Gamma \models \sigma = \tau : \Delta &:\Iff &&& %
    \Gamma \der \sigma = \tau : \Delta \And \Gamma \models \And \Delta
    \models  \And \\
    &&&&& %
    \all I, \delta,\gamma (I \force \delta = \gamma : \Gamma \Imp I
    \force \sigma \delta = \tau \gamma : \Delta)
    \\
    & \Gamma \models \sigma : \Delta &:\Iff &&&
    \Gamma \der \sigma : \Delta \And
    \Gamma \models \sigma = \sigma : \Delta
  \end{align*}
\end{definition}

\begin{remark}
  \label{rem:def-snd}
  \leavevmode
  \begin{enumerate}
  \item For each $I$ we have ${}\force I$, and $J \force \sigma : I$
    if{f} $\sigma \co J \to I$; likewise, $J \force \sigma = \tau : I$
    if{f} $\sigma = \tau$.
  \item For computability of contexts and substitutions monotonicity
    and partial equivalence properties hold analogous to computability
    of types and terms.
  \item\label{item:fla-val} Given $\force \Gamma$ and $I \force \sigma =
    \tau : \Gamma$, then for any $\Gamma \der \phi : \FF$ we get $\phi
    \sigma = \phi \tau \in \FF(I)$ since $\phi \sigma$ and $\phi \tau$
    only depend on the name assignments of $\sigma$ and $\tau$ which
    have to agree by $I \force \sigma = \tau : \Gamma$.  Similarly for
    $\Gamma \der r : \II$.
  \item The definition of ``$\models$'' slightly deviates from the
    approach we had in the definition of ``$\force$'' as, say, $\Gamma
    \models A$ is defined in terms of $\Gamma \models A = A$.  Note
    that by the properties we already established about ``$\force$''
    we get that $\Gamma \models A = B$ implies $\Gamma \models A$ and
    $\Gamma \models B$ (given we know $\Gamma \der A$ and $\Gamma \der
    B$, respectively); and, likewise, $\Gamma \models a = b : A$
    entails $\Gamma \models a : A$ and $\Gamma \models b : A$ (given
    $\Gamma \der a : A$ and $\Gamma \der b : A$, respectively).  Also,
    note that in the definition of, say, $\Gamma \models A$, the
    condition
    \[
    \all I, \sigma,\tau (I \force \sigma = \tau : \Gamma \Imp I \force
    A \sigma = A \tau)
    \]
    implies
    \[
    \all I, \sigma (I \force \sigma : \Gamma \Imp I \force A \sigma).
    \]
    In fact, we will often have to establish the latter condition
    first when showing the former.
  \item $I \models A = B$ if{f} $I \force A = B$, and $I \models a = b
    : A$ if{f} $I \force A$ and $I \force a = b : A$; moreover, given
    $I \models A$ and $I, x : A \der B$, then $I, x: A \models B$
    if{f}
    \begin{align*}
      &\all f \co J \to I \all u (J \force u : A f \Imp J \force B (f,
      \su x u )) \And\\
      &\all f \co J \to I \all u,v (J \force u = v : A f \Imp J \force
      B (f,\su x u ) = B(f,\su x v )) %
    \end{align*}
    (Note that the second formula in the above display implies the
    first.)  Thus the premises of \rcpi\ and \rcsi\ are simply $I
    \models A$ and $I, x :A \models B$.  Also, $I,\phi \force A = B$
    if{f} $I,\phi \models A = B$; and $I,\phi \models a = b : A$ if{f}
    $I,\phi \force A$ and $I,\phi \force a = b : A$.
  \item By Lemma~\ref{lem:eq-resp-trans} we get that $\Gamma \models
    \cdot = \cdot$, $\Gamma \models \cdot = \cdot : A$, and $\Gamma
    \models \cdot = \cdot : \Delta$ are partial equivalence relations.
  \end{enumerate}
\end{remark}

\begin{theorem}[Soundness]
  \label{thm:snd}
  $\Gamma \der \J \Imp \Gamma \models \J$
\end{theorem}
The proof of the Soundness Theorem spans the rest of this section.  We
will mainly state and prove congruence rules as the proof of the other
rules are special cases.

\begin{lemma}
  \label{lem:ctxt-snd}
  The context formation rules are sound:
  \begin{mathpar}
    \inferrule { } {\emptyctx \models {}} %
    \and %
    \inferrule {\Gamma \models \\ i \notin \dom(\Gamma)} {\Gamma, i :
      \II \models {}} %
    \and %
    \inferrule {\Gamma \models \phi : \FF } {\Gamma, \phi \models {}}
    \and %
    \inferrule {\Gamma \models A \\ x \notin \dom (\Gamma)} {\Gamma, x
      : A \models} %
  \end{mathpar}
\end{lemma}
\begin{proof}
  Immediately by definition.
\end{proof}

\begin{lemma}
  \label{lem:FF-II-snd}
  Given $\Gamma \models$, $\Gamma \der r : \II$, $\Gamma \der s : \II$,
  $\Gamma \der \phi : \FF$, and $\Gamma \der \psi : \FF$ we have:
  \begin{enumerate}
  \item\label{item:9} $\Gamma \der r = s : \II \Imp \Gamma \models r =
    s : \II$
  \item\label{item:15} $\Gamma \der \phi = \psi : \FF \Imp\Gamma
    \models \phi = \psi : \FF$
  \end{enumerate}
\end{lemma}
\begin{proof}
  \eqref{item:9} By virtue of
  Remark~\ref{rem:def-snd}~\eqref{item:fla-val} it is enough to show
  $r \sigma = s \sigma \in \II(I)$ for $I \force \sigma : \Gamma$.
  But then by applying the substitution $I \der \sigma : \Gamma$ we
  get $I \der r\sigma = s \sigma : \II$, and thus $r \sigma = s \sigma
  \in \II(I)$ since the context $I$ does not contain restrictions.
  The proof of \eqref{item:15} is analogous.
\end{proof}

\begin{lemma}
  \label{lem:type-conv-snd}
  The rule for type conversion is sound:
  \begin{mathpar}
    \inferrule {\Gamma \models a = b : A \\
      \Gamma \models A = B } %
    { \Gamma \models a = b : B}
  \end{mathpar}
\end{lemma}
\begin{proof}
  Suppose $I \force \sigma = \tau : \Gamma$. By assumption we have $I
  \force a \sigma = b \tau : A \sigma$.  Moreover also $I \force
  \sigma = \sigma : \Gamma$, so $I \force A \sigma = B \sigma$, and
  hence $I \force a \sigma = b \tau : B \sigma$ by
  Lemma~\ref{lem:eq-resp-trans} which was what we had to prove.
\end{proof}
\begin{lemma}
  \label{lem:subst-snd}
  \begin{mathpar}
    \inferrule { %
      \Gamma \models \sigma = \tau : \Delta \\
      \Delta \models A = B
    }%
    { \Gamma \models A \sigma = B \tau } %
    \and %
    \inferrule { %
      \Gamma \models \sigma = \tau : \Delta \\
      \Delta \models a = b : A
    }%
    { \Gamma \models a \sigma = b \tau : A \sigma } %
    \and %
    \inferrule { %
      \Gamma \models \sigma = \tau : \Delta \\
      \Delta \models \delta = \gamma : \Xi
    }%
    { \Gamma \models \delta \sigma  = \gamma \tau : \Xi } %
  \end{mathpar}
\end{lemma}
\begin{proof}
  Immediate by definition.
\end{proof}

\begin{lemma}
  \label{lem:pi-snd}
  The rules for $\Pi$-types are sound:
  \begin{enumerate}
    \openup 2\jot
  \item \label{eq:17}
    $\inferrule { %
      \Gamma \models A = A' \\
      \Gamma, x : A \models B = B' } %
    { \Gamma \models (x : A) \to B = (x : A') \to B' }$
  \item %
    \label{eq:18} $\inferrule { 
      \Gamma \models A = A' \\
      \Gamma, x : A \models t = t' : B } %
    { \Gamma \models \lambda x : A. t = \lambda x : A'. t' : (x : A)
      \to B } $ %
    \item %
    \label{eq:19}
    $\inferrule { %
      \Gamma \models w = w' : (x : A) \to B \\
      \Gamma \models u = u' : A } %
    { \Gamma \models w \, u = w' \, u' : B \subst x u }$ %
    \item %
    \label{eq:20}
    $\inferrule { %
      \Gamma, x : A \models t : B \\
      \Gamma \models u : A } %
    { \Gamma \models (\lambda x : A. t) \, u = t \subst x u : B \subst
      x u }$ %
    \item %
    \label{eq:21}
    $\inferrule { %
      \Gamma \models w : (x : A) \to B \\
      \Gamma \models w' : (x : A) \to B \\
      \Gamma, x : A \models w \, x = w' \, x : B} %
    { \Gamma \models w = w' : (x : A) \to B }$
  \end{enumerate}
\end{lemma}
\begin{proof} Abbreviate $(x: A) \to B$ by $C$.  We will make use of
  Remark~\ref{rem:transitivity-and-pi}.

  \eqref{eq:17} It is enough to prove this in the case where $\Gamma$
  is of the form $I$, in which case this directly follows by \repi.

  \eqref{eq:18}
  %
  Suppose $\Gamma \models A = A'$ and $\Gamma, x : A \models t = t' :
  B$; this entails $\Gamma, x: A \models B$.  For $I \force \sigma =
  \tau : \Gamma$ we show $I \force (\lambda x : A. t) \sigma =
  (\lambda x : A'. t') \tau : C \sigma$.  For this let $J \force u = v
  : A \sigma f$ where $f \co J \to I$.  Then also $J \force u = v : A'
  \tau f$,
  \begin{align*}
    &J \der (\lambda x : A. t) \sigma f \, u \sred t (\sigma f, \su x
    u ) : B (\sigma f, \su x u ), \text{ and}
    \\
    &J \der (\lambda x : A'. t') \tau f \, v \sred t' (\tau f, \su x
    v ) : B (\tau f, \su x v).
  \end{align*}
  Moreover, $J \force (\sigma f, \su x u ) = (\tau f, \su x v ):
  \Gamma, x :A$, and so $J \force B(\sigma f, \su x u ) = B(\tau f,
  \su x v )$
  and
  \[
  J \force t (\sigma f, \su x u ) = t' (\tau f, \su x v ) : B (\sigma
  f, \su x u )
  \]
  which gives
  \begin{align*}
    &J \force (\lambda x : A. t) \sigma f \, u = t (\sigma f, \su x u
    ) : B (\sigma f, \su x u ), \text{ and}
    \\
    &J \force (\lambda x : A'. t') \tau f \, v = t' (\tau f, \su x v )
    : B (\tau f, \su x v),
  \end{align*}
  by applying the Expansion Lemma twice, and thus also
  \[
  J \force (\lambda x : A. t) \sigma f \, u = (\lambda x : A'. t')
  \tau f \, v : B (\sigma f, \su x u )
  \]
  what we had to show.

  \eqref{eq:19} For $I \force \sigma = \tau : \Gamma$ we get $I \force
  w \sigma = w' \tau : C \sigma$ and $I \force u \sigma = u' \tau : A
  \sigma$; so also $I \force w \sigma: C \sigma$, therefore $I \force
  (w \, u) \sigma = w \sigma \, u'\tau = (w'\, u') \tau : B (\sigma,
  \su x u)$.

  \eqref{eq:20} Given $I \force \sigma = \tau : \Delta$ we get, like
  in \eqref{eq:18}, $I \force (\lambda x : A. t) \sigma \, u \sigma =
  t (\sigma, \su x {u\sigma}) : B (\sigma, \su x {u \sigma})$ using
  the Expansion Lemma; moreover, $I \force (\sigma, \su x {u \sigma})
  = (\tau, \su x {u \tau}) : \Gamma, x : A$, hence
  \[
  I \force (\lambda x : A. t) \sigma \, u \sigma = t (\sigma, \su x
  {u\sigma}) = t (\tau, \su x {u \tau}): B (\sigma, \su x {u
    \sigma}).
  \]

  \eqref{eq:21} Suppose $I \force \sigma = \tau : \Gamma$ and $J
  \force u : A \sigma f$ for $f \co J \to I$.  We have to show $J
  \force w \sigma f \, u = w' \tau f \, u : B(\sigma f, \su x u )$.
  We have
  \[
  J \force (\sigma f, \su x u ) = (\tau f, \su x u ): \Gamma,x:A
  \]
  and thus, by the assumption $\Gamma, x : A \models w \, x = w' \, x
  : B$, we get
  \[
  J \force (w \, x) (\sigma f, \su x u ) = (w'\, x) (\tau f, \su x u)
  : B (\sigma f, \su x u ).
  \]
  Since $x$ does neither appear in $w$ nor in $w'$ this was what we
  had to prove.
\end{proof}

\begin{lemma}
  \label{lem:sigma-snd}
  The rules for $\Sigma$-types are sound:
  \begin{enumerate}
    \openup 2\jot
  \item \label{eq:sigma-intro} $\inferrule { \Gamma \models A = A' \\
      \Gamma, x : A \models B = B' } %
    {\Gamma \models (x : A) \times B = (x : A') \times B'}$
  \item \label{eq:pair-intro} $\inferrule { %
      \Gamma, x : A \models B \\
      \Gamma \models u = u' : A\\
      \Gamma \models v = v' : B \subst x u\\
    } %
    { \Gamma \models (u,v) = (u',v') : (x : A) \times B }$
  \item \label{eq:proj}
    $\inferrule { %
      \Gamma, x : A \models B \\
      \Gamma \models w = w' : (x : A) \times B}
    { \Gamma \models w.1 = w'.1 : A\\\\
      \Gamma \models w.2 =  w'.2 : B \subst x {w.1} }$
  \item \label{eq:pair-beta} $\inferrule { %
      \Gamma, x : A \models B \\
      \Gamma \models u : A \\
      \Gamma \models v : B \subst x u } %
    { \Gamma \models (u,v).1 = u : A \\\\
      \Gamma \models (u,v).2 = v : B \subst x u}$
  \item \label{eq:pair-eta}
    $\inferrule { %
      \Gamma, x : A \models B \\
      \Gamma \models w : (x : A) \times B\\
      \Gamma \models w' : (x : A) \times B\\
      \Gamma \models w.1 = w'.1 : A\\
      \Gamma \models w.2 =  w'.2 : B \subst x {w.1} }
    { \Gamma \models w = w' : (x : A) \times B}$
  \end{enumerate}
\end{lemma}

\begin{lemma}
  \label{lem:natrec-is-computable}
  Given $I, x : \N \models C$ we have:
  \begin{enumerate}
    \openup 2\jot
  \item \label{eq:natrec-comp}
    $\inferrule { %
      I \force u : \N \\
      I \force z : C \subst x 0\\
      I \force s : (x : \N) \limp C \limp C \subst x {\suc x}\\
    } %
    { I \force \natrec \, u \, z \, s : C \subst x u \\\\
      I \force \natrec \, u \, z \, s = (\natrec \, u \, z \, s) \dn :
      C \subst x u }$
  \item \label{eq:natrec-eq}
    $\inferrule { %
      I \force u = u' : \N \\
      I \force z = z' : C \subst x 0\\
      I \force s = s' : (x : \N) \limp C \limp C \subst x {\suc x}\\
    } %
    { I \force \natrec \, u \, z \, s = \natrec \, u' \, z' \, s' : C
      \subst x u}$
  \end{enumerate}
\end{lemma}
\begin{proof}
  By simultaneous induction on $I \force u : \N$ and $I \force u = u'
  : \N$.

  \case\ $I \force 0 : \N$. We have $I \der \natrec \, 0 \, z \, s
  \sred z : C \subst x 0$ so~\eqref{eq:natrec-comp} follows from the
  Expansion Lemma.

  \case\ $I \force 0 = 0 : \N$. \eqref{eq:natrec-eq} immediately follows
  from~\eqref{eq:natrec-comp} and $I \force z = z' : C \subst x 0$.

  \case\ $I \force \suc u : \N$ from $I \force u : \N$.  We have
  \[
  I \der \natrec \, (\suc u) \, z \, s \sred s \, u \, (\natrec \, u
  \, z \, s) : C \subst x {\suc u}
  \]
  and $I \force s \, u \, (\natrec \, u \, z \, s) : C \subst x {\suc
    u}$ by \IH, and using that $u$ and $s$ are computable.  Hence we
  are done by the Expansion Lemma.

  \case\ $I \force \suc u = \suc u' : \N$ from $I \force u = u' : \N$.
  \eqref{eq:natrec-eq} follows from~\eqref{eq:natrec-comp} and $ I
  \force s = s' : (x : \N) \limp C \limp C \subst x {\suc x}$, $I
  \force u = u' : \N$, and the \IH.

  \case\ $I \force u : \N$ for $u$ non-introduced.  For $f \co J \to
  I$ we have
  \[
  J \der (\natrec \, u \, z \, s)f \red \natrec \, (u f \dn) \, zf \,
  sf : C (f, \su x {uf \dn}).
  \]
  Moreover, we have $I \force u f \dn$ and $I \force u f \dn = u \dn f
  : \N$ with a shorter derivation (and thus also $J \force C (f, \su x
  {uf \dn}) = C \subst x {u \dn} f$), hence by \IH
  \begin{align*}
    J &\force \natrec \, (u f \dn) \, zf \, sf : C \subst x {u \dn} f,
    \text{
      and}\\
    J &\force \natrec \, (u f \dn) \, zf \, sf = (\natrec \, (u \dn)
    \, z \, s) f : C \subst x {u \dn} f,
  \end{align*}
  which yields the claim by the Expansion Lemma.

  \case\ $I \force u = u' : \N$ for $u$ or $u'$ non-introduced.  We
  have
  \[
  I \force \natrec \, u \, z \, s = \natrec \, (u \dn) \,
  z \, s : C \subst x u
  \]
  by either \eqref{eq:natrec-comp} (if $u$ is non-introduced) or by
  reflexivity (if $u$ is an introduction); likewise for $u'$.  So with
  the \IH\ for $I \force u \dn = u' \dn : \N$ we obtain
  \begin{equation*}
    I \force \natrec \, u \, z \, s = \natrec \, (u \dn) \, z \, s %
    = \natrec \, (u' \dn) \, z' \, s' = \natrec \, u' \, z' \, s' : C
    \subst x u
  \end{equation*}
  what we had to show.
\end{proof}

We write $\num n$ for the numeral $\suc^n 0$ where $n \in \NN$.

\begin{lemma}
  \label{lem:n-canon}
  If $I \force u : \N$, then $I \force u = \num n : \N$ (and hence
  also $I \der u = \num n : \N$) for some $n \in \NN$.
\end{lemma}
\begin{proof}
  By induction on $I \force u : \N$.  The cases for zero and successor
  are immediate.  In case $u$ is non-introduced, then $I \force u \dn
  = \num n$ for some $n \in \NN$ by \IH.  By
  Lemma~\ref{lem:red-eq}~\eqref{item:nat-red-cl} and transitivity
  we conclude $I \force u = \num n : \N$.
\end{proof}

\begin{lemma}
  \label{lem:n-discrete}
  $I \force \cdot = \cdot : \N$ is discrete, i.e., if $I \force u : \N$, $I
  \force v : \N$, and $J \force u f = v g : \N$ for some $f, g \co J \to
  I$, then $I \force u = v : \N$.
\end{lemma}
\begin{proof}
  By Lemma~\ref{lem:n-canon}, we have $I \force u = \num n : \N$ and
  $I \force v = \num m : \N$ for some $n,m \in \NN$, and thus $J
  \force \num n = u f = v g = \num m : \N$, i.e., $J \force \num n =
  \num m : \N$ and hence $n = m$ which yields $I \force u = v : \N$.
\end{proof}

\begin{lemma}
  \label{lem:path-sound}
  The rules for $\Path$-types are sound:
  \begin{enumerate}
    \openup 2\jot
  \item \label{eq:15} %
    $\inferrule { %
      \Gamma \models A = A' \\
      \Gamma \models u = u' :A\\
      \Gamma \models v = v' :A } %
    { \Gamma \models \Path\, A \, u \, v = \Path \, A' \, u' \, v'}$
  \item \label{eq:nabs-intro-snd} %
    $\inferrule { %
      \Gamma \models A \\
      \Gamma, i:\II \models t = t' : A } %
    {\Gamma \models \nabs i t = \nabs i t' : \Path\,A\,t(i0)\,t(i1)}$
  \item \label{eq:path-app-snd} %
    $\inferrule { %
      \Gamma \models w = w' : \Path\,A\,u\,v \\
      \Gamma \models r = r' : \II } %
    { \Gamma \models w \, r = w' \, r' : A}$
  \item \label{eq:path-app-endpoints-snd} %
    $\inferrule { %
      \Gamma \models w : \Path\,A\,u\,v\\
    } %
    { \Gamma \models w \, 0 = u : A\\
      \Gamma \models w \, 1 = v : A}$
  \item \label{eq:path-beta-snd} %
    $\inferrule { %
      \Gamma \models A\\
      \Gamma, i : \II \models t : A\\
      \Gamma \models r : \II
    } %
    { \Gamma \models (\nabs i t)\, r = t \subst i r : A}$
  \item \label{eq:path-eta-snd} %
    $\inferrule { %
      \Gamma \models w : \Path\,A\,u\,v \\
      \Gamma \models w' : \Path\,A\,u\,v \\
      \Gamma, i : \II \models w \, i = w' \, i : A\
    } %
    { \Gamma \models w = w' : \Path\,A\,u\,v \\}$
  \end{enumerate}
\end{lemma}
\begin{proof}
  \eqref{eq:15} Follows easily by definition.

  \eqref{eq:nabs-intro-snd} For $I \force \sigma = \sigma' : \Gamma$
  we have to show
  \begin{equation}
    \label{eq:34}
    I \force (\nabs i t) \sigma = (\nabs i t') \sigma' : \Path\,A
    \sigma\,t(\sigma,\su i 0 )\,t(\sigma,\su i 1 ).
  \end{equation}
  For $f \co J \to I$ and $r \in \II (J)$ we have $J \force (\sigma f,
  \su i r ) = (\sigma' f, \su i r ) : \Gamma, i : \II$ and
  \begin{align*}
    J &\der (\nabs i t) (\sigma f) \, r \sred t (\sigma f, \su i r ) :
    A \sigma f, \text{ and}\\
    J &\der (\nabs i t') (\sigma' f) \, r \sred t' (\sigma' f, \su i r
    ) : A \sigma' f,
  \end{align*}
  and moreover $J \force t (\sigma f, \su i r ) = t' (\sigma' f, \su i
  r ) : A \sigma f$ and $J \force A \sigma f = A \sigma' f$ by
  assumption.  Hence the Expansion Lemma yields
  \begin{align*}
    J &\force (\nabs i t) (\sigma f) \, r = t (\sigma f, \su i r ) :
    A \sigma f, \text{ and}\\
    J &\force (\nabs i t') (\sigma' f) \, r = t' (\sigma' f, \su i r )
    : A \sigma f,
  \end{align*}
  in particular also, say $J \force (\nabs i t) \sigma \, 0 = t
  (\sigma, \su i 0 ) : A \sigma$ and $J \force (\nabs i t') \sigma' \,
  0 = t' (\sigma', \su i 0 ) = t (\sigma, \su i 0 ) : A \sigma$.  And
  hence \eqref{eq:34} follows.

  \eqref{eq:path-app-snd} Supposing $I \force \sigma = \sigma' :
  \Gamma$ we have to show $I \force (w \sigma) \, (r \sigma) =
  (w'\sigma') \, (r' \sigma) : A \sigma$.  We have $I \force w \sigma
  = w'\sigma' : \Path\,A \sigma\,u \sigma\,v \sigma$ and $r \sigma =
  r' \sigma'$, hence the claim follows by definition.

  \eqref{eq:path-app-endpoints-snd} Let $I \force \sigma = \sigma' :
  \Gamma$; we have to show, say, $I \force w \sigma \, 0 = u \sigma' :
  A \sigma$.  First, we get $I \force w \sigma : \Path\,A \sigma\,
  u\sigma \, v \sigma$. Since $\Gamma \models w : \Path\,A\,u\,v$ we
  also have $\Gamma \models \Path\,A\,u\,v$, hence
  \begin{equation}
    \label{eq:35}
    I \force \Path\,A \sigma \, u \sigma \, v\sigma = \Path\,A
    \sigma' \, u \sigma' \, v \sigma'.
  \end{equation}
  Hence we also obtain $I \force w \sigma : \Path\,A \sigma'\, u
  \sigma'\, v \sigma'$, and thus $I \force w \sigma \, 0 = u \sigma' :
  A \sigma'$.  But~\eqref{eq:35} also yields $I \force A \sigma = A
  \sigma'$ by definition, so $I \force w \sigma \, 0 = u \sigma' : A
  \sigma$ what we had to show.

  \eqref{eq:path-beta-snd} Similar to \eqref{eq:nabs-intro-snd} using
  the Expansion Lemma.

  \eqref{eq:path-eta-snd} For $I \force \sigma = \sigma' : \Gamma$, $f
  \co J \to I$, and $r \in \II (J)$, we have $J \force (\sigma f, \su
  i r ) = (\sigma'f,\su i r ) : \Gamma,i : \II$, and thus
  \begin{equation}
    \label{eq:36}
    J \force (w \, i)  (\sigma f, \su i r ) =  (w' \, i)
    (\sigma'f,\su i r )  : A \sigma f.
  \end{equation}
  But $(w \, i) (\sigma f, \su i r )$ is $w \sigma f ~ r$, and $(w' \,
  i) (\sigma'f, \su i r )$ is $w' \sigma' f ~ r$, so~\eqref{eq:36} is
  what we had to show.
\end{proof}

\begin{lemma}
  \label{lem:systems}
  Let $\phi_i \in \FF(I)$ and $\phi_1 \lor \dots \lor \phi_n = 1$.
  \begin{enumerate}
  \item\label{item:sys-typ} Let $I, \phi_i \lforce A_i$ and $I, \phi_i
    \land \phi_j \lforce A_i = A_j$ for all $i,j$; then
    \begin{enumerate}
    \item $I \lforce [ \phi_1 \smap A_1, \dots, \phi_n \smap A_n]$,
      and
    \item $I \lforce [ \phi_1 \smap A_1, \dots, \phi_n \smap A_n] =
      A_k$ whenever $\phi_k = 1$.
    \end{enumerate}
  \item\label{item:sys-ter} Let $I \lforce A$, $I,\phi_i \lforce t_i :
    A$, and $I,\phi_i \land \phi_j \lforce t_i = t_j : A$ for all
    $i,j$; then
    \begin{enumerate}
    \item $I \lforce [ \phi_1 \smap t_1, \dots, \phi_n \smap t_n] :
      A$, and
    \item $I \lforce [ \phi_1 \smap t_1, \dots, \phi_n \smap t_n] =
      t_k : A$ whenever $\phi_k = 1$.
    \end{enumerate}
  \end{enumerate}
\end{lemma}
\begin{proof}
  \eqref{item:sys-typ} Let us abbreviate $[ \phi_1 \smap A_1, \dots,
  \phi_n \smap A_n]$ by $A$.  Since $A$ is non-introduced, we have to
  show $J \force A f \dn$ and $J \force A f \dn = A \dn f$.  For the
  former observe that $A f \dn$ is $A_k f$ with $k$ minimal such that
  $\phi_k f = 1$.  For the latter use that $J \force A_k f = A_l f$ if
  $\phi_k f = 1$ and $\phi_l = 1$, since $I,\phi_k \land \phi_l \force
  A_k = A_l$.

  \eqref{item:sys-ter} Let us write $t$ for $[ \phi_1 \smap t_1,
  \dots, \phi_n \smap t_n]$.  By virtue of the Expansion Lemma, it
  suffices to show $J \force t f \dn : A f$ and $K \force t f \dn = t
  \dn f : A f$.  The proof is just like the proof for types given
  above.
\end{proof}

\begin{lemma}
  \label{lem:sys-cover}
  Given $\Gamma \models \phi_1 \lor \dots \lor \phi_n = 1 : \FF$,
  then:
  \begin{mathpar}
    \inferrule { \Gamma, \phi_1 \models \J  ~\dots~  \Gamma, \phi_n
      \models \J }
    { \Gamma \models \J}
  \end{mathpar}
\end{lemma}
\begin{proof}
  Let $\phi = \phi_1 \lor \dots \lor \phi_n$.  Say if $\J$ is a typing
  judgment of the form $A$.  For $I \force \sigma : \Gamma$ we have
  $\phi \sigma = 1$, so $\phi_k \sigma = 1$ for some $k$, hence $I
  \force A \sigma$ by $\Gamma,\phi_k \models A$.  Now let $I \force
  \sigma = \tau : \Gamma$; then $\phi_i \sigma = \phi_i \tau$
  ($\sigma$ and $\tau$ assign the same elements to the interval
  variables), so $\phi \sigma = \phi \tau = 1$ yields $\phi_k \sigma =
  \phi_k \tau = 1$ for some common $k$ and thus $I \force A \sigma = A
  \tau$ follows from $\Gamma, \phi_k \models A$.  The other judgment
  forms are similar.
\end{proof}

For $I \force A$ and $I, \phi \force v : A$ we write $I \force u : A
[\phi \mapsto v]$ for $I \force u : A$ and $I,\phi \force u = v : A$.
And likewise $I \force u = w : A [\phi \mapsto v]$ means $I \force u =
w : A$ and $I,\phi \force u = v : A$ (in this case also $I,\phi \force
w = v : A$ follows).  We use similar notations for for ``$\models$''.

\begin{lemma}
  \label{lem:Glue-snd-local}
  Given $\phi \in \FF (I)$ and $I \lforce A, I,\phi \lforce T$, and
  $I,\phi\lforce w : \Equiv\,T\,A$, and write $B$ for $\Glue\,[\phi
  \mapsto (T,w)]\,A$.  Then:
  \begin{enumerate}
  \item\label{item:locGlue-intro} $I \lforce B$ and $I,\phi \lforce B =
    T$.
  \item\label{item:locGlue-eq} If $I \lforce A = A'$, $I, \phi \lforce
    T = T'$, $I, \phi \lforce w = w' : \Equiv\,T\,A$, then $I \lforce
    B = \Glue\,[\phi \mapsto (T',w')]\,A'$.
  \item\label{item:locunglue-comp} If $I \lforce u : B$ and $I, \phi
    \lforce w = w' : \Equiv\,T\,A$, then $I \lforce
    \unglue\,[\phi\mapsto w']\, u : A [\phi \mapsto w'.1 \, u]$ and $I
    \lforce \unglue\,[\phi\mapsto w]\,u = \unglue\,[\phi\mapsto w']\,u
    : A$.
  \item\label{item:locunglue-eq} If $I \lforce u = u' : B$, then
    \[
    I \lforce \unglue\,[\phi\mapsto w]\,u = \unglue\,[\phi\mapsto
    w]\,u' : A.
    \]
  \item\label{item:locglue} If $I, \phi \lforce t = t' : T$ and $I \lforce
    a = a' : A[\phi \mapsto w.1 \, t]$, then
    \begin{enumerate}
    \item $I \lforce \glue\,[\phi \mapsto t]\,a = \glue\,[\phi \mapsto
      t']\,a' : B$,
    \item\label{item:locglue-res} $I,\phi \lforce \glue\,[\phi \mapsto
      t]\,a = t : T$, and
    \item\label{item:locunglue-beta} $I \lforce \unglue\,[\phi\mapsto
      w]\, (\glue\,[\phi \mapsto t]\,a) = a : A$.
    \end{enumerate}
  \item\label{item:locunglue-eta} If $I \lforce u : B$, then $I
    \lforce u = \glue\,[\phi \mapsto u] (\unglue\,[\phi\mapsto w]\,u)
    : B$.
  \end{enumerate}
\end{lemma}
\begin{proof}
  %
  \eqref{item:locGlue-intro} Let us first prove $I,\phi \force B$ and
  $I,\phi \force B = T$; but in $I,\phi$, $\phi$ becomes $1$ so
  w.l.o.g.\ let us assume $\phi = 1$; then $B$ is non-introduced and $I
  \der B \sred T$ so $I \force B$ from $I \force T$.  For $I \force B
  = T$ we have to show $J \force B f \dn = T f \dn$ for $f \co J \to
  I$.  But $B f \dn$ is $T f$ so this is an instance of
  Lemma~\ref{lem:red-eq}.

  It remains to prove $I \force B$ in case where $\phi \neq 1$; for
  this use \rcgl\ with the already proven $I,\phi \force B$.

  \eqref{item:locGlue-eq} In case $\phi \neq 1$ we only have to show
  $I,\phi \force B = B'$ and can apply \regl.  But restricted to
  $I,\phi$, $\phi$ becomes $1$ and hence we only have to prove the
  statement for $\phi = 1$.  But then by \eqref{item:locGlue-intro} we
  have $I \force B = T = T' = B'$.

  \eqref{item:locunglue-comp} In case $\phi \neq 1$, $I \force
  \unglue\,[\phi\mapsto w']\, u : A$ and
  \begin{equation}
    \label{eq:33}
    I \lforce \unglue\,[\phi\mapsto w]\,u = %
    \unglue\,[\phi\mapsto w']\,u : A
  \end{equation}
  are immediate by definition.  Using the Expansion Lemma (and $I \der
  \unglue\,[\phi\mapsto w']\,u \sred w'.1 \, u : A$ for $\phi = 1$) we
  obtain $I,\phi \force \unglue\,[\phi\mapsto w']\,u = w'.1 \, u : A$,
  which also shows $I \force \unglue\,[\phi\mapsto w']\,u : A$ as well
  as \eqref{eq:33} in case $\phi = 1$.

  \eqref{item:locunglue-eq} In case $\phi \neq 1$, this is by
  definition. For $\phi = 1$ we have
  \[
  I \force \unglue\,[\phi\mapsto w]\,u = w.1 \, u = w.1 \, u' =
  \unglue\,[\phi\mapsto w]\,u' : A.
  \]

  \eqref{item:locglue} Let us write $b$ for $\glue\,[\phi \mapsto
  t]\,a$, and $b'$ for $\glue\,[\phi \mapsto t']\,a'$.  We first show
  $I \force b : B$ and $I,\phi \force b= t : B$ (similarly for $b'$).

  In case $\phi = 1$, $I \der b \sred t : T$ so by the Expansion Lemma
  $I \force b : T$ and $I \force b = t : T$, and hence also $I \force
  b : B$ and $I \force b = t : B$ by \eqref{item:locGlue-intro}.  This
  also proves \eqref{item:locglue-res}.

  Let now $\phi$ be arbitrary; we claim
  \begin{equation*}
    I \force \unglue\,[\phi \mapsto w]\,b : A \text{ and } %
    I \force \unglue\,[\phi \mapsto w]\,b = a : A
  \end{equation*}
  (and thus proving \eqref{item:locunglue-beta}).  We will apply the
  Expansion Lemma to do so; for $f \co J \to I$ let us analyze the
  reduct of $(\unglue\,[\phi \mapsto w]\,b) f$:
  \begin{equation*}
    (\unglue\,[\phi \mapsto w]\, b) f \dn =
    \begin{cases}
      w f.1 \, bf & \text{if }\phi f = 1,\\
      a f & \text{otherwise.}
    \end{cases}
  \end{equation*}
  Note that, if $\phi f = 1$, we have as in the case for $\phi =1$, $J
  \force b f = t f : B f$ and hence $J \force wf.1 \, b f = wf.1 \, t
  f = af : A f$.  This ensures $J \force (\unglue\,[\phi\mapsto w]\,b)
  f \dn = (\unglue\,[\phi\mapsto w]\,b) \dn f : A$, and thus the
  Expansion Lemma applies and we obtain $I \force
  \unglue\,[\phi\mapsto w]\,b = (\unglue\,[\phi\mapsto w]\,b) \dn :
  A$; but as we have seen in either case, $\phi = 1$ or not, $I \force
  (\unglue \,[\phi\mapsto w]\, b) \dn = a : A$ proving the claim.

  Let now be $\phi \neq 1$, $f \co J \to I$, and $J \force w'=wf :
  \Equiv\,Tf\,Af$.  We can use the claim for $B f$ and $\Glue\,[\phi f
  \mapsto (Tf,w')]\,Af$ (which is forced equal to $B f$
  by~\eqref{item:locGlue-eq}) and obtain both
  \begin{gather*}
    J \force \unglue\,[\phi f \mapsto wf]\,bf = af : Af \text{ and }
    J \force \unglue\,[\phi f \mapsto w']\,bf = af : Af,
  \end{gather*}
  so the left-hand sides are equal; moreover, $I,\phi \force b : B$
  (as in the case $\phi =1$), and hence $I \force b : B$.  Likewise
  one shows $I \force b' : B$.



  It remains to show $I \force b = b' : B$.  If $\phi =1$, we already
  showed $I \force b = t : T$ and $I \force b' = t' : T$, so the claim
  follows from $I \force t = t' : T$ and $I \force T = B$.  Let us now
  assume $\phi \neq 1$.  We immediately get $I, \phi \force b = t = t'
  = b' : B$ as for $\phi =1$.  Moreover, we showed above that $I
  \force \unglue\,[\phi \mapsto w]\,b = a : A$ and $I \force
  \unglue\,[\phi \mapsto w]\,b' = a' : A$.  Hence we obtain
  \[
  I \force \unglue\,[\phi \mapsto w]\, b = \unglue\,[\phi \mapsto w]\,
  b' : A
  \]
  from $I \force a = a' : A$.


  \eqref{item:locunglue-eta} In case $\phi = 1$, this follows from
  \eqref{item:locglue-res}.  In case $\phi \neq 1$, we have to show
  \begin{gather*}
    I \force \unglue\,[\phi \mapsto w]\,u = \unglue\,[\phi \mapsto
    w]\,(\glue\,[\phi \mapsto u] (\unglue\,[\phi \mapsto w]\,u))
    : A \text{ and }\\
    I,\phi \force u = \glue\,[\phi \mapsto u] (\unglue\,[\phi \mapsto
    w]\,u) : T.
  \end{gather*}
  The former is an instance of \eqref{item:locunglue-beta}; the latter
  follows from~\eqref{item:locglue-res}.
\end{proof}

\begin{lemma}
  \label{lem:Glue-param-size}
  Let $B$ be $\Glue \, [ \phi \mapsto (T, w)] \, A$ and suppose $I
  \force B$ is derived via \rcgl, then also $I,\phi \force T$ and the
  derivations of $I,\phi \force T$ are all proper sub-derivations of
  $I \force B$ (and hence shorter).
\end{lemma}
\begin{proof}
  We have the proper sub-derivations $I,\phi \force B$.  For each $f
  \co J \to I$ with $\phi f = 1$, we have that $B f$ is non-introduced
  with reduct $T f$ so the derivation of $J \force B f$ has a
  derivation of $J \force T f$ as sub-derivation according to \rcni.
\end{proof}

For the next proof we need a small syntactic observation.  Given
$\Gamma \der \alpha : \FF$ irreducible, there is an associated
substitution $\bar \alpha \co \Gamma_\alpha \to \Gamma$ where
$\Gamma_\alpha$ skips the names of $\alpha$ and applies a
corresponding $\bar \alpha$ to the types and restrictions (e.g., if
$\Gamma$ is $i:\II, x : A, j : \II, \phi$ and $\alpha$ is $(i=0)$,
then $\Gamma_\alpha$ is $x : A(i0),j:\II,\phi (i0)$).  Since $\alpha
\bar \alpha = 1$ we even have $\bar \alpha \co \Gamma_\alpha \to
\Gamma, \alpha$.  The latter has an inverse (w.r.t.\ judgmental
equality) given by the projection $\pp \co \Gamma,\alpha \to
\Gamma_\alpha$ (i.e., $\pp$ assigns each variable in $\Gamma_\alpha$
to itself): in the context $\Gamma,\alpha$, $\bar \alpha \pp$ is the
identity, and $\pp \bar \alpha$ is the identity since the variables in
$\Gamma_\alpha$ are not changed by $\bar \alpha$.

\begin{remark}
  \label{rem:restricted-contexts}
  We can use the above observation to show that the condition $I, \phi
  \der \J$ in the definition of $I,\phi \lforce \J$ (in
  Section~\ref{sec:computability}) already follows from the other,
  i.e., $J \lforce \J f$ for all $f \co J \to I,\phi$: We have to show
  $I,\alpha \der \J$ for each irreducible $\alpha \leq \phi$.  But we
  have $I_\alpha \lforce \J \bar\alpha$ by the assumption and
  $\bar\alpha \co I_\alpha \to I,\phi$, and hence $I_\alpha \der \J
  \bar\alpha$.  Substituting along $\pp \co I,\alpha \to I_\alpha$
  yields $I,\alpha \der \J$.
\end{remark}

\begin{theorem}
  \label{thm:comp-is-computable}
  Compositions are computable, i.e., for $\phi \in \FF (I)$ and $i
  \notin \dom (I)$:
  \begin{enumerate}
    \openup 2\jot
  \item \label{eq:comp-comp} %
    $\inferrule { %
      I,i \force A \\
      I,i,\phi \force u : A \\
      I \force u_0 : A (i0) [\phi \mapsto u(i0)] } %
    { I \force \comp^i \, A \, [\phi \mapsto u] \, u_0 : A (i1) [\phi
      \mapsto u(i1)] \\\\
      I \force \comp^i \, A \, [\phi \mapsto u] \, u_0 = %
      (\comp^i \, A \, [\phi \mapsto u] \, u_0) \dn : A(i1)}$
  \item \label{eq:comp-eq} %
    $\inferrule { %
      I,i \force A \\
      I,i,\phi \force u = v : A \\
      I \force u_0 = v_0: A (i0) [\phi \mapsto u(i0)]
    } %
    { I \force \comp^i \, A \, [\phi \mapsto u] \, u_0 = %
      \comp^i \, A \, [\phi \mapsto v] \, v_0: A (i1)
    }$
  \item \label{eq:comp-type-eq} %
    $\inferrule { %
      I,i \force A = B \\
      I,i,\phi \force u : A \\
      I \force u_0 : A (i0) [\phi \mapsto u(i0)]
    } %
    { I \force \comp^i \, A \, [\phi \mapsto u] \, u_0 = %
      \comp^i \, B \, [\phi \mapsto u] \, u_0: A (i1)
    }$
  \end{enumerate}
\end{theorem}
\begin{proof}
  By simultaneous induction on $I,i \force A$ and $I,i \force A = B$.
  Let us abbreviate $\comp^i \, A \, [\phi \mapsto u] \, u_0$ by
  $u_1$, and $\comp^i \, A \, [\phi \mapsto v] \, v_0$ by $v_1$.  The
  second conclusion of~\eqref{eq:comp-comp} holds since in each case
  we will use the Expansion Lemma and in particular also prove $I
  \force u_1 \dn : A (i1)$.

  Let us first make some preliminary remarks.  Given the induction
  hypothesis holds for $I,i \force A$ we also know that filling
  operations are admissible for $I,i \force A$, i.e.:
  \begin{gather}
    \label{eq:fillcomp}
    \inferrule { %
      I,i \force A \\
      I,i,\phi \force u : A \\
      I \force u_0 : A (i0) [\phi \mapsto u(i0)] } %
    { I,i \force \Comp^i \, A \, [\phi \mapsto u] \, u_0 : A [\phi
      \mapsto u, (i=1) \mapsto u_1] }
  \end{gather}
  To see this, recall the explicit definition of filling
  \[
  \Comp^i\,A\,[\varphi\mapsto u]\,u_0 = \comp^j\,A\subst{i}{i\land
    j}\,[\phi \mapsto u\subst{i}{i \land j}, (i=0) \mapsto u_0]\,u_0
  \]
  where $j$ is fresh.  The derivation of $I,i,j \force A \subst{i}{i
    \land j}$ isn't higher than the derivation of $I,i \force A$ so we
  have to check, with $u' = [\phi \smap u\subst{i}{i \land j}, (i=0)
  \smap u_0]$ and $A' = A \subst{i}{i \land j}$,
  \begin{equation}
    \label{eq:14}
    I,i,j,\phi \lor (i=0) \force u' : A' \text { and }I,i,\phi\lor(i=0)
    \force u'(j0) =  u_0 : A(i0).
  \end{equation}
  To check the former, we have to show
  \begin{equation*}
    I,i,j,\phi \land (i=0) \force u \subst i {i \land j} = u_0 : A'
  \end{equation*}
  in order to apply Lemma~\ref{lem:systems}.  So let $f \co J \to
  I,i,j$ with $\phi f = 1$ and $f (i) = 0$; then as $\phi$ doesn't
  contain $i$ and $j$, also $\phi (f-i,j) = 1$ for $f - i,j \co J \to
  I$ being the restriction of $f$, so by assumption $J \force u(i0)
  (f-i,j) = u_0 (f -i,j) : A(i0) (f-i,j)$.  Clearly, $(i0) (f-i,j) =
  \subst i {i \land j} f$ so the claim follows.

  Let us now check the right-hand side equation of \eqref{eq:14}:
  by virtue of Lemma~\ref{lem:sys-cover} we have to check the equation
  in the contexts $I,i,\phi$ and $I,i,(i=0)$; but $I,i,\phi \force
  u'(j0) = u(i0) = u_0 : A(i0)$ and $I, i, (i=0) \force u'(j0) = u_0 :
  A(i0)$ by Lemma~\ref{lem:systems}.

  And likewise the filling operation preserves equality.
  \\

  \case\ \rcn. First, we prove that
  \begin{equation}
    \label{eq:5}
    I,\phi,i: \II \der u = u_0 : \N.
  \end{equation}
  To show this, it is enough to prove $I,\alpha,i: \II \der u = u_0 :
  \N$ for each $\alpha \leq \phi$ irreducible.  Let $\bar \alpha \co
  I_\alpha \to I$ be the associated face substitution.  We have
  $I_\alpha, i \force u (\bar \alpha, \su i i ) : \N$ and also
  $I_\alpha \force u (\bar \alpha, \su i 0 ) = u_0 \bar \alpha: \N$
  since $\phi \bar \alpha = 1$.  By discreteness of $\N$
  (Lemma~\ref{lem:n-discrete}),
  \[
  I_\alpha, i \force u (\bar \alpha, \su i i ) = u_0 \bar \alpha : \N,
  \]
  therefore $I_\alpha, i \der u (\bar \alpha, \su i i ) = u_0 \bar
  \alpha : \N$, i.e., $I_\alpha, i \der u \bar \alpha = u_0 \bar
  \alpha : \N$ with $\bar \alpha$ considered as substitution
  $I_\alpha,i \to I,i$ and $u_0$ weakened to $I,i$. Hence $I, \alpha,
  i : \II \der u = u_0 : \N$ by the observation preceding the
  statement of the theorem.

  Second, we prove that
  \begin{equation}
    \label{eq:6}
    I,\phi \force u(i1) = u_0 : \N.
  \end{equation}
  $I,\phi \der u(i1) = u_0 : \N$ immediately follows from
  \eqref{eq:5}.  For $f\co J \to I$ with $\phi f = 1$ we have to show
  $J \force u (i1) f = u_0 f :\N$; since $\phi f = 1$ we get $J \force
  u (i0) f = u_0 f :\N$ by assumption, i.e., $J \force u (f,\su i j
  )(j0) = u_0 (f,\su i j ) (j0) :\N$ (where $u_0$ is weakened to $I,j$
  and $j$ fresh).  By discreteness of $\N$, we obtain $J,j \force u
  (f,\su i j ) = u_0 (f,\su i j ) :\N$ and hence $J \force u (f,\su i
  1 ) = u_0 (f,\su i 1 ) :\N$, i.e., $J \force u (i1) f = u_0 f : \N$.

  We now prove the statements simultaneously by a side induction on $I
  \force u_0 : \N$ and $I \force u_0 = v_0 : \N$.

  \subcase\ $I \force 0 : \N$.  By~\eqref{eq:5} it follows that $I
  \der u_1 \sred 0 : \N$, and hence $I \force u_1 : \N$ and $I \force
  u_1 = 0 : \N$ by the Expansion Lemma.  Thus also $I,\phi \force u_1
  = u(i1) : \N$ by~\eqref{eq:6}.

  \subcase\ $I \force \suc u_0' : \N$ from $I \force u_0' : \N$ with
  $u_0 = \suc u_0'$.  By \eqref{eq:5} it follows that
  \[
  I \der u_1 \sred \suc (\comp^i \, \N \, [\phi \mapsto \pred u] \, u'_0)
  : \N.
  \]

  From $I,\phi \force \suc u_0' = u (i0) : \N$ we get $I,\phi \force
  u'_0 = \pred (\suc u_0') = \pred u(i0) : \N$ using
  Lemma~\ref{lem:natrec-is-computable} and thus by \SIH, $I \force
  \comp^i \, \N \, [\phi \mapsto \pred u] \, u'_0 : \N [\phi \mapsto
  (\pred u) (i1)]$; hence $I \force u_1 : \N$ and $I \force u_1 = \suc
  ( \comp^i \, \N \, [\phi \mapsto \pred u] \, u'_0) : \N$ by the
  Expansion Lemma.  Thus also
  \[
  I,\phi \force u_1 = \suc (\pred u (i1)) = \suc (\pred (\suc u_0')) =
  \suc u_0' = u(i1) : \N
  \]
  using \eqref{eq:6}.

  \subcase\ $u_0$ is non-introduced.  We use the Expansion Lemma: for
  each $f \co J \to I$
  \[
  u_1 f \dn = \comp^j \, \N \, [\phi f \mapsto u (f,\su i j )] \, (u_0
  f \dn)
  \]
  the right-hand side is computable by \SIH, and this results in a
  compatible family of reducts by \SIH, since we have $K \force u_0 f
  \dn g = u_0 f g \dn : \N$.  Thus we get $I \force u_1 : \N$ and $I
  \force u_1 = u_1 \dn : \N$.  By \SIH, $I,\phi \force u_1 \dn = u(i1)
  : \N$ and thus also $I,\phi \force u_1 = u(i1) : \N$.

  \subcase\ $I \force 0 = 0 : \N$.  Like above we get that $I \force
  u_1 = 0 = v_1 : \N$.

  \subcase\ $I \force S u_0' = S v_0' : \N$ from $I \force u_0' = v_0'
  : \N$.  Follows from the \SIH\ $I \force \comp^i \, \N \, [\phi
  \mapsto \pred u] \, u'_0 = \comp^i \, \N \, [\phi \mapsto \pred v]
  \, v'_0 : \N$ like above.

  \subcase\ $I \force u_0 = v_0 : \N$ and $u_0$ or $v_0$ is
  non-introduced.  We have to show $J \force u_1 f \dn = v_1 f \dn :
  \N$ for $f \co J \to I$.  We have $J \force u_0 f \dn = v_0 f \dn :
  \N$ with a shorter derivation, thus by \SIH
  \[
  J \force \comp^j \, \N \, [\phi f \mapsto u (f,\su i j )] \, (u_0 f
  \dn) = \comp^j \, \N \, [\phi f \mapsto v (f,\su i j )] \, (v_0 f
  \dn) : \N
  \]
  which is what we had to show.

  \case\ \rcpi.  Let us write $(x : A) \to B$ for the type under
  consideration. \eqref{eq:comp-comp} In view of the Expansion Lemma,
  the reduction rule for composition at $\Pi$-types (which is closed
  under substitution), and Lemma~\ref{lem:pi-snd} \eqref{eq:18} and
  \eqref{eq:21}, it suffices to show
  \begin{align}
    \label{eq:26}
    I, x : A(i1) &\models \comp^i \, B \subst{x}{\bar x} \, [\phi
    \mapsto u \, \bar x] \, (u_0 \, \bar x (i0)): B(i1), \text{ and}
    \\
    \label{eq:27}
    I, x : A(i1),\phi &\models \comp^i \, B \subst{x}{\bar x} \, [\phi
    \mapsto u \, \bar x] \, (u_0 \, \bar x (i0)) = u (i1) \, x: B(i1),
  \end{align}
  where $x' = \Comp^i\,A \subst{i} {1-i}\,[]\,x$ and $\bar x = x'
  \subst{i}{1-i}$.  By \IH, we get $I, x : A(i1), i:\II \models \bar x
  : A $ and $I, x : A(i1) \models \bar x (i1) = x : A(i1)$, i.e.,
  \begin{align}
    \label{eq:23}
    I, x : A(i1), i : \II &\models \Comp^i\,A \subst{i} {1-i}\,[]\,x :
    A \subst i {1-i}, \text{ and}
    \\
    \label{eq:28}
    I, x : A(i1) &\models (\Comp^i\,A \subst{i} {1-i}\,[]\,x)(i0) = x:
    A(i1).
  \end{align}
  To see \eqref{eq:23}, let $J \force (f,\su x a ) = (f,\su x b ) : I,
  x : A(i1)$, i.e., $f \co J \to I$ and $J \force a = b : A (i1) f$;
  for $j$ fresh, we have $J,j \force A (f,\su i {1-j})$ (note that
  $(i1) f = (f, \su i {1-j}) (j0)$) and we get
  \[
  J,j \force \Comp^j\,A (f,\su i {1-j})\,[]\,a = \Comp^j\,A (f,\su i
  {1-j})\,[]\,b : A (f,\su i {1-j})
  \]
  by \IH, i.e., $J,j \force x' (f,\su x a ,\su i j ) = x' (f,\su x a
  ,\su i j ) : A (f, \su i {1 -j})$, and hence for $r \in \II (J)$
  \[
  J \force x' (f,\su x a ,\su i r ) = x' (f, \su x b ,\su i r ) : (A
  \subst i {1-i}) (f,\su x a ,\su i r ).
  \]

  Thus we get $I, x : A(i1) \models u_0 \, \bar x (i0) : B (i0) \subst
  x {\bar x (i0)}$, $I,x : A(i1),\phi, i : \II \models u \, \bar x : B
  \subst x {\bar x}$, and
  \[
  I,x: : A(i1), \phi \models u_0 \, \bar x (i0) = u (i0)\, \bar x (i0)
  = (u \, \bar x) (i0) : B (i0) \subst x {\bar x (i0)}.
  \]
  And hence again by \IH, we obtain \eqref{eq:26} and \eqref{eq:27}.

  \eqref{eq:comp-eq} Let $f \co J \to I$ and $J \force a:A(f,\su i 1
  )$.  Then $J,j \force \bar a : A(f,\su i j )$ as above and we have to
  show
  \begin{multline}
    \label{eq:31}
    J \force %
    \comp^j \,B(f,\su x {\bar a},\su i j )\,[\phi f \mapsto u (f,\su i
    j ) \, \bar a] \,(u_0f\, \bar a) \\
    = \comp^j \,B(f,\su x {\bar a}, \su i j )\,[\phi f \mapsto v
    (f,\su i j ) \, \bar a] \,(v_0f\, \bar a):\\
    B(f,\su x {\bar a (i1)}, \su i 1 ).
  \end{multline}
  But this follows directly from the \IH\ for $J,j \force B(f,\su x
  {\bar a}, \su i {j})$.

  \case\ \rcsi. Let us write $(x : A) \times B$ for the type under
  consideration. \eqref{eq:comp-comp}~We have
  \[
  I,i,\phi \force u.1 : A \quad \text{and} \quad I \force u_0.1 : A [\phi
  \mapsto u.1]
  \]
  so by \IH,
  \[
  I,i \force \Comp^i \, A \, [\phi \mapsto u.1] \, (u_0.1) : A [\phi
  \mapsto u.1, (i=0) \mapsto u_0.1].
  \]
  Let us call the above filler $w$.  Thus we get $I,i \force B \subst
  x w$,
  \[ I,i,\phi \force B \subst x {u.1} = B \subst x w \quad \text{and}
  \quad %
  I \force B \subst x {u_0.1} = (B \subst x w) (i0)
  \]
  and hence
  \[
  I,i,\phi \force u.2 : B \subst x w \quad \text{and} \quad
  I \force u_0.1 : (B \subst x w) (i0) [\phi \mapsto u.2].
  \]
  The \IH\ yields
  \[
  I \force \comp^i \, B\subst x w \, [\phi \mapsto u.2] \, (u_0.2) :
  (B \subst x w) (i1) [\phi \mapsto u.2 (i1)];
  \]
  let us write $w'$ for the above.  By the reduction rules for
  composition in $\Sigma$-types we get $I \der u_1 \sred (w(i1),w') :
  (x : A(i1)) \times B (i1)$ and hence the Expansion Lemma yields
  \[
  I \force u_1 = (w(i1),w') : (x : A(i1)) \times B (i1).
  \]
  Which in turn implies the equality
  \[
  I,\phi \force u_1 = (w(i1),w') = (u.1 (i1), u.2 (i1)) = u (i1) :
  (x : A(i1)) \times B (i1).
  \]

  The proof of \eqref{eq:comp-eq} uses that all notions defining $w$
  and $w'$ preserve equality (by \IH), and thus $I \force u_1 \dn =
  v_1 \dn : (x : A(i1)) \times B(i1)$.

  \case\ \rcpa.  Let us write $\Path\,A\,a_0\,a_1$ for the type under
  consideration.  We obtain (for $j$ fresh)
  \begin{multline}
    \label{eq:29}
    I, j \force \comp^i \, A \,[ (j=0) \mapsto a_0, (j=1) \mapsto a_1,
    \phi \mapsto u \, j] \, (u_0 \, j): \\
    A(i1) [ (j=0) \mapsto a_0 (i1), (j=1) \mapsto a_1(i1), \phi
    \mapsto u (i1) \, j]
  \end{multline}
  by the \IH. Using the Expansion Lemma, the reduction rule for
  composition at $\Path$-types, and
  Lemma~\ref{lem:path-sound}~\eqref{eq:nabs-intro-snd} this yields
  \[
  I \force u_1 : \Path\,A(i1)\,\tilde u (j0) \,\tilde u (j1) [\phi \mapsto
  \nabs j {(u (i1) \, j)}]
  \]
  where $\tilde u$ is the element in \eqref{eq:29} and $u_1$ is $\nabs
  j \tilde u$.  But $I \force \tilde u (jb) = a_b (i1) : A(i1)$, so $I
  \force u_1 : \Path\,A(i1)\,a_0(i1)\,a_1 (i1)$.  Moreover,
  \[
  I,\phi \force u_1 = \nabs j (u(i1) \, j) = u(i1) :
  \Path\,A(i1)\,a_0(i1)\,a_1 (i1)
  \]
  by the correctness of the $\eta$-rule for paths
  (Lemma~\ref{lem:path-sound}~\eqref{eq:path-eta-snd}).

  \case\ \rcgl.  To not confuse with our previous notations, we write
  $\psi$ for the face formula of $u$, and write $B$ for $\Glue \, [
  \phi \mapsto (T, w)] \, A$.

  Thus we are given:
  \[
  \nir {\rcgl} { %
    1 \neq \phi \in \FF (I,i) \\
    I,i \force A \\
    I,i, \phi \force w : \Equiv \, T\, A \\
    I,i, \phi \force B } %
  { I,i \force B}
  \]
  and also $I,i,\psi \force u : B$ and $I \force u_0 : B (i0) [\psi
  \mapsto u(i0)]$.  Moreover we have $I,i,\phi \force T$ with shorter
  derivations by Lemma~\ref{lem:Glue-param-size}.  We have to show
  \begin{enumerate}[(i)]
  \item\label{item:gluecomp} $I \force u_1 : B(i1)$, and
  \item\label{item:gluecompcorr} $I,\psi \force u_1 = u(i1) : B(i1)$.
  \end{enumerate}
  We will be using the Expansion Lemma: let $f \co J \to I$ and
  consider the reducts of $u_1 f$:
  \begin{equation*}
    u_1 f \dn =
    \begin{cases}
      \comp^j \,T f' \, [\psi f \mapsto u f']\,(u_0 f)
      & \text{if }\phi f' = 1,\\
      \glue\, [\phi (i1) f \mapsto t_1 f]\,(a_1f)& \text{otherwise,}
    \end{cases}
  \end{equation*}
  with $f'=(f,\su i {j})$, and $t_1$ and $a_1$ as in the corresponding
  reduction rule, i.e.:
  \begin{align*}
    a &= \unglue\,[\phi\mapsto w]\, u
    && I,i,\psi
    \\
    a_0 &= \unglue\,[\phi (i0)\mapsto w(i0)]\, u_0
    && I
    \\
    \delta &= \forall i.\varphi && I
    \\
    a_1' &= \comp^i\,A\,[\psi\mapsto a]\,a_0 && I
    \\
    t_1' &= \comp^i\,T\,[\psi\mapsto u]\,u_0 && I,\delta
    \\
    \omega &=\pres^i\,w\,[\psi\mapsto u]\,u_0 && I,\delta
    \\
    (t_1,\alpha) &= \eq\,w(i1)\,[\delta \mapsto (t'_1,\omega),
    \psi \mapsto (u(i1),\nabs {j}{a_1'})]\,a_1' && I,\varphi(i1)
    \\
    a_1 &= \comp^j\,A(i1)\,[\varphi(i1)\mapsto \alpha\,j,\psi\mapsto
    a(i1)]\,a_1' && I
  \end{align*}

  First, we have to check $J \force u_1 f \dn : B (i1) f$.  In case
  $\phi f' = 1$ this immediately follows from the \IH.  In case $\phi
  f' \neq 1$, this follows from the \IH\ and the previous lemmas
  ensuring that notions involved in the definition of $t_1$ and $a_1$
  preserve computability.

  Second, we have to check $J \force u_1 f \dn = u_1 \dn f : B (i1)
  f$.  For this, the only interesting case is when $\phi f' = 1$; then
  we have to check that:
  \begin{equation}
    \label{eq:glue-comp-unif}
    J \force \comp^j \,T f' \, [\psi f \mapsto u f']\,(u_0 f) = %
    \glue\, [\phi (i1) f \mapsto t_1 f]\,(a_1f) : B (i1) f
  \end{equation}
  Since all the involved notions commute with substitutions, we may
  (temporarily) assume $f = \id$ and $\phi = 1$ to simplify notation.
  Then also $\delta = 1 = \phi (i1)$, and hence (using the \IH)
  \begin{align*}
    I \force t_1 &= t_1' = \comp^i \, T\, [\psi \mapsto u] \, u_0 : T
    (i1),
  \end{align*}
  so \eqref{eq:glue-comp-unif} follows from
  Lemma~\ref{lem:Glue-snd-local}~\eqref{item:locglue-res} and
  \eqref{item:locGlue-intro}.

  So the Expansion Lemma yields \eqref{item:gluecomp} and $I \force
  u_1 = \glue\, [\phi (i1) \mapsto t_1]\,a_1 : B(i1)$.
  \eqref{item:gluecompcorr} is checked similarly to what is done in
  \cite[Appendix~A]{CohenCoquandHuberMortberg15} using the \IH.  This
  proves~\eqref{eq:comp-comp} in this case; for \eqref{eq:comp-eq} one
  uses that all notions for giving $a_1$ and $t_1$ above preserve
  equality, and thus $I \force u_1 \dn = v_1 \dn : B(i1)$ entailing $I
  \force u_1 = v_1 : B(i1)$.

  \case\ \rcu.  We have
  \[
  I \der \comp^i \, \U \, [\phi \mapsto u] \, u_0 \sred \Glue\,[\phi
  \mapsto (u(i1), \ptoeq^i {u \subst {i} {1-i}})]\,u : \U
  \]
  thus it is sufficient to prove that the right-hand side is
  computable, i.e.,
  \begin{equation*}
    I \force_1 \Glue\,[\phi \mapsto (u(i1), \ptoeq^i {u \subst {i}
      {1-i}})]\,u_0 : \U
  \end{equation*}
  that is,
  \begin{equation*}
    I \force_0 \Glue\,[\phi \mapsto (u(i1), \ptoeq^i {u \subst {i}
      {1-i}})]\,u_0.
  \end{equation*}
  We have $I \force_0 u_0$ so by
  Lemma~\ref{lem:Glue-snd-local}~\eqref{item:locGlue-intro} it
  suffices to prove
  \begin{equation*}
    I \force_0 \ptoeq^i {u \subst {i} {1-i}} : \Equiv\,u(i1)\,u_0.
  \end{equation*}
  To see this recall that the definition of $\ptoeq^i {u \subst {i}
    {1-i}}$ is defined from compositions and filling operations for
  types $I,i \force_0 u$ and $I,i \force_0 u \subst i {1-i}$ using
  operations we already have shown to preserve computability.  But in
  this case we have as \IH, that these composition and filling
  operations are computable since the derivations of $I,i \force_0 u$
  and $I,i \force_0 u$ are less complex than the derivation $I
  \force_1 \U$ since the level is smaller.

  \case\ \rcni.  So we have $J \force A f \dn$ for each $f \co J \to
  I,i$ and $J \force A \dn f = A f \dn$ (all with a shorter derivation
  than $I,i \force A$).  Note that by
  Lemma~\ref{lem:red-eq}~\eqref{item:8}, we also have $I,i \force A =
  A \dn$.

  \eqref{eq:comp-comp} We have to show $J \force u_1 f : A (i1) f \dn$
  for each $f \co J \to I$.  It is enough to show this for $f$ being
  the identity; we do this using the Expansion Lemma.  Let $f \co J
  \to I$ and $j$ be fresh, $f' = (f,\su i {j})$; we first show $J
  \force u_1 f \dn : A \dn (i1) f$.  We have
  \[
  J \der u_1 f \red \comp^j \, (A f' \dn) \, [\phi f \mapsto u f'] \,
  u_0 f : A f' (j1)
  \]
  hence also at type $A f' (j1) \dn$, and so, by
  \IH~\eqref{eq:comp-comp} for $J,j \force A f' \dn$, we obtain $J
  \force u_1 f \dn : A f' (j1) \dn$.  But $J \force A f' (j1)\dn = A
  \dn (i1) f$, so $J \force u_1 f \dn : A \dn (i1) f$.

  Next, we have to show $J \force u_1 \dn f = u_1 f \dn : A \dn (i1)
  f$.  Since $J,j \force A \dn f' = A f' \dn$ (with a shorter
  derivation) we get by \IH~\eqref{eq:comp-type-eq}, $J \force u_1 \dn
  f = u_1 f \dn : A \dn f' (j1)$ what we had to show.

  Thus we can apply the Expansion Lemma and obtain $I \force u_1 : A
  \dn (i1)$ and $I \force u_1 = u_1 \dn : A \dn (i1)$, and hence also
  $I \force u_1 : A (i1)$ and $I \force u_1 = u_1 \dn : A (i1)$.  By
  \IH, we also have $I,\phi \force u_1 = u_1 \dn = u(i1) : A \dn (i1)
  = A (i1)$.

  \eqref{eq:comp-eq} Like above, we obtain
  \begin{equation*}
    I \force u_1 = u_1 \dn : A \dn (i1) \quad \text {and} \quad %
    I \force v_1 = v_1 \dn : A \dn (i1).
  \end{equation*}
  But since the derivation of $I,i \force A \dn$ is shorter, and $u_1
  \dn = \comp^i \, A\dn \, [\phi \mapsto u] \, u_0$ and similarly for
  $v_1 \dn$, the \IH\ yields $I \force u_1 \dn = v_1 \dn : A \dn
  (i1)$, thus also $I \force u_1 = v_1 : A \dn (i1)$, that is, $I
  \force u_1 = v_1 : A (i1)$ since $I,i \force A = A \dn$.


  It remains to show that composition preserves forced type equality
  (i.e.,~\eqref{eq:comp-type-eq} holds).  The argument for the
  different cases is very similar, namely using that the compositions
  on the left-hand and right-hand side of \eqref{eq:comp-type-eq} are
  equal to their respective reducts (by \eqref{eq:comp-comp}) and then
  applying the \IH\ for the reducts.  We will only present the case
  \reni.


  \case\ \reni.  Then $A$ or $B$ is non-introduced and $I,i \force A
  \dn = B \dn$ with a shorter derivation.  Moreover, by
  \eqref{eq:comp-comp} (if the type is non-introduced) or reflexivity
  (if the type is introduced) we have
  \begin{align*}
    I &\force \comp^i \, A \, [\phi \mapsto u] \, u_0 = %
    \comp^i \, (A \dn) \, [\phi \mapsto u] \, u_0  : A(i1), \text{ and}
    \\
    I &\force \comp^i \, B \, [\phi \mapsto u] \, u_0 = %
    \comp^i \, (B  \dn) \, [\phi \mapsto u] \, u_0 : B(i1),
  \end{align*}
  but the right-hand sides are forced equal by \IH.
\end{proof}

\begin{lemma}
  \label{lem:univ-sound}
  The rules for the universe $\U$ are sound:
  \begin{enumerate}
  \item $\Gamma \models A : \U \Imp \Gamma \models A$
  \item\label{item:10} $\Gamma \models A = B : \U \Imp \Gamma \models
    A = B$
  \end{enumerate}
  Moreover, the rules reflecting the type formers in $\U$ are sound.
\end{lemma}
\begin{proof}
  Of the first two statements let us only prove~\eqref{item:10}: given
  $I \force \sigma = \tau : \Gamma$ we get $I \force A\sigma = B\tau :
  \U$; this must be a derivation of $I \force_1 A\sigma = B\tau : \U$
  and hence we also have $I \force_0 A\sigma = B\tau$.

  The soundness of the rules reflecting the type formers in $\U$ is
  proved very similar to proving the soundness of the type formers.
  Let us exemplify this by showing soundness for $\Pi$-types in $\U$:
  we are give $\Gamma \models A : \U$ and $\Gamma, x : A \models B :
  \U$, and want to show $\Gamma \models (x : A) \to B: \U$.  Let $I
  \force \sigma = \tau : \Gamma$, then $I \force A \sigma = A \tau :
  \U$, so, as above, $I \force_0 A \sigma = A \tau$; it is enough to
  show
  \begin{equation}
    \label{eq:9}
    J \force_0 B (\sigma f, \su x u ) = B (\tau f, \su x v )
  \end{equation}
  for $J \force u = v : A \sigma f$ with $f \co J \to I$.  Then $J
  \force (\sigma f, \su x u ) = (\tau f, \su x v ) : \Gamma, x : A$,
  hence $J \force B (\sigma f, \su x {u}) = B (\tau f, \su x {v}) :
  \U$ and hence~\eqref{eq:9}.
\end{proof}

\begin{proof}[Proof of Soundness (Theorem~\ref{thm:snd})]
  By induction on the derivation $\Gamma \der \J$.

  We have already seen above that most of the rules are sound.  Let us
  now look at the missing rules.  Concerning basic type theory, the
  formation and introduction rules for $\N$ are immediate; its
  elimination rule and definitional equality follow from the ``local''
  soundness from Lemma~\ref{lem:natrec-is-computable} as follows.
  Suppose $\Gamma \models u : \N$, $\Gamma, x : \N \models C$, $\Gamma
  \models z : C \subst x 0$, and $\Gamma \models s : (x : \N) \limp C
  \limp C \subst x {\suc x}$.  For $I \force \sigma = \tau : \N$ we
  get by Lemma~\ref{lem:natrec-is-computable}~\eqref{eq:natrec-eq}
  \[
  I \force \natrec\,u\sigma\,z\sigma\,s\sigma =
  \natrec\,u\tau\,z\tau\,s\tau : C (\sigma, \su x {u \sigma}).
  \]
  (Hence $\Gamma \models \natrec\,u\,z\,s : C \subst x u$.)
  Concerning, the definitional equality, if, say, $u$ was of the form
  $\suc v$, then,
  Lemma~\ref{lem:natrec-is-computable}~\eqref{eq:natrec-comp} gives
  \[
  I \force \natrec\,(\suc v\sigma)\,z\sigma\,s\sigma = \natrec\,(\suc
  v\tau)\,z\tau\,s\tau = (\natrec\,(\suc v\tau)\,z\tau\,s\tau)\dn : C
  (\sigma, \su x {u \sigma}).
  \]
  and $(\natrec\,(\suc v\tau)\,z\tau\,s\tau)\dn$ is $s \tau \,
  v\tau\,(\natrec\,v\tau\,z\tau\,s\tau)$, proving
  \[
  \Gamma \models \natrec\,(\suc v)\,z\,s = s\,v\,(\natrec\,v\,z\,s): C
  \subst x {\suc v};
  \]
  similarly, the soundness of the other definitional equality is
  established.

  Let us now look at the composition operations: suppose $\Gamma, i :
  \II \models A$, $\Gamma \models \phi : \FF$, $\Gamma,\phi, i : \II
  \models u : A$, and $\Gamma \models u_0 : A (i0) [\phi \mapsto u
  (i0)]$.  Further let $I \force \sigma = \tau : \Gamma$, then for $j$
  fresh, $I, j \force \sigma' = \tau' : \Gamma,i:\II$ where $\sigma' =
  (\sigma,\su i {j})$ and $\tau' = (\tau,\su i {j})$, hence $I,j
  \force A \sigma' = A \tau'$, $\phi \sigma = \phi \tau$, $I,j,\phi
  \sigma \force u \sigma' = u \tau' : A \sigma'$, and $I \force u_0
  \sigma = u_0 \tau : A \sigma' (j0) [\phi \sigma \mapsto u \sigma'
  (j0)]$.  By Theorem~\ref{thm:comp-is-computable},
  \[
  I \force \comp^j\,(A\sigma')\,[\phi \sigma \mapsto u \sigma']\, (u_0
  \sigma) =%
  \comp^j\,(A\tau')\,[\phi \tau \mapsto u \tau']\, (u_0 \tau) : A
  \sigma' (j1)
  \]
  and
  \[
  I,\phi \sigma \force \comp^j\,(A\sigma')\,[\phi \sigma \mapsto u
  \sigma']\, (u_0 \sigma) %
  = u \sigma' (j1) = u \tau' (j1) : A \sigma' (j1)
  \]
  hence we showed $\Gamma \models \comp^i\,A\,[\phi \mapsto u]\,u_0 :
  A (i1) [\phi \mapsto u (i1)]$.  Similarly one can justify the
  congruence rule for composition.

  The definitional equalities which hold for $\comp$ follow from the
  second conclusion of
  Theorem~\ref{thm:comp-is-computable}~\eqref{eq:comp-comp}, i.e.,
  that a composition is forced equal to its reduct.

  The remaining rules for systems follow from their ``local''
  analogues in form of Lemma~\ref{lem:systems}; let us, say, suppose
  $\Gamma \models \phi_1 \lor \dots \lor \phi_n = 1 : \FF$,
  $\Gamma,\phi_i \models A_i$, and $\Gamma,\phi_i \land \phi_j \models
  A_i = A_j$.  For $I \force \sigma = \tau : \Gamma$ we get $k$ with
  $\phi_k \sigma = \phi_k \tau = 1$ like in the proof of
  Lemma~\ref{lem:sys-cover} so, writing $A$ for $[ \phi_1 \smap A_1,
  \dots, \phi_n \smap A_n]$,
  \[
  I \force A \sigma = A_k \sigma = A_k \tau = A \tau
  \]
  by Lemma~\ref{lem:systems} and using $\Gamma,\phi_k \models A_k$,
  so $\Gamma \models A$.  Likewise, if $\Gamma \models \phi_l = 1
  : \FF$ for some $l$, then $I \force A \sigma = A_l \sigma = A_l
  \tau$, showing $\Gamma \models A = A_l$ in this case.  The other
  rules concerning systems are justified similarly.

  The soundness of the remaining rules concerning $\Glue$ follow
  similarly from their ``local'' version in
  Lemma~\ref{lem:Glue-snd-local}.
\end{proof}

\begin{corollary}[Canonicity]
  \label{cor:canonicity}
  If $I$ is a context of the form $i_1 : \II, \dots, i_k : \II$ and $I
  \der u : \N$, then $I \der u = \num n : \N$ for a unique $n \in \NN$.
\end{corollary}
\begin{proof}
  By Soundness, $I \models u : \N$ hence $I \force u : \N$, so $I
  \force u = \num n : \N$ for some $n \in \NN$ by
  Lemma~\ref{lem:n-canon}, and thus also $I \der u = \num n : \N$.
  The uniqueness follows since $I \der \num n = \num m : \N$ yields $I
  \force \num n = \num m : \N$ which is only the case for $n = m$.
\end{proof}

\begin{corollary}[Consistency]
  \label{cor:consistency}
  Cubical type theory is consistent, i.e., there is a type in the
  empty context which is not inhabited.
\end{corollary}
\begin{proof}
  Consider the type $\Path\,\N\,0\,1$ and suppose there is a $u$ with
  $\der u : \Path\,\N\,0\,1$.  Hence we get $i : \II \der u \, i :
  \N$, as well as $\der u \, 0 = 0 : \N$ and $\der u \, 1 = 1 : \N$.
  By Canonicity, we get $n \in \NN$ with $i : \II \der u \, i = \num n
  : \N$, and hence (by substitution) $\der u\,0 = \num n : \N$ and
  $\der u\,1 = \num n : \N$, so $\der 0 = 1 : \N$, contradicting the
  uniqueness in Corollary~\ref{cor:canonicity}.
\end{proof}

\begin{remark}
  \label{rem:consistency}
  One could also extend cubical type theory with an empty type $\N_0$
  whose forcing relation is empty; consistency for this extension is
  then an immediate consequence of the corresponding Soundness
  Theorem.
\end{remark}

\begin{remark}
  \label{rem:pi-injective-u-canon}
  Soundness also implies injectivity of $\Pi$ (and likewise for other
  type formers) in name contexts: if $I \der (x : A) \to B = (x : A')
  \to B'$, then $I \der A = A'$ and $I, x : A \der B = B'$.  Moreover,
  we get a canonicity result for the universe $\U$: if $I \der A :
  \U$, then $A$ is judgmentally equal to an introduced type $B$ with
  $I \force_0 B$.
\end{remark}

\section{Extension with Higher Inductive Types}
\label{sec:hits}

In this section we discuss two extensions to cubical type theory with
two higher inductive types: the circle and propositional
truncation. For both extensions it is suitable to generalize path
types to dependent path types $\Path^i\,A\,u\,v$ where $i$ might now
appear in $A$, with $u$ in $A(i0)$ and $v$ in $A(i1)$.  This extension
is straightforward, e.g., the $\beta$-reduction rule for paths now
reads
\begin{mathpar}
  \inferrule { %
    \Gamma, i : \II \der A \\
    \Gamma, i : \II \der t : A\\
    \Gamma \der r : \II}%
  { \Gamma \der (\nabs i t) \, r \red t \subst i r : A \subst i r}
\end{mathpar}
and likewise the computability predicates and relations are easily
adapted.

\subsection{The Circle}
\label{sec:circle}

In this section we sketch how the proof of canonicity can be extended
to the system where a circle $\sone$ is added; the extension with
$n$-spheres is done analogously.


First, we have to extend the reduction relation as follows to
incorporate the circle.
\begin{mathparpagebreakable}
  \inferrule { %
    \Gamma \der \\
  } %
  { \Gamma \der \sloop 0 \red \base : \sone \\\\
    \Gamma \der \sloop 1 \red \base : \sone
  } %
  \and %
  \inferrule { %
    \Gamma, i: \II \der u : \sone\\
  } %
  {
    \Gamma \der \comp^i\,\sone\,[1 \mapsto u ]\,u(i0) \red u(i1) : \sone
  } %
\end{mathparpagebreakable}
(For simplicity, we will use $\comp^i\,\sone$ instead of adding yet
another constructor $\hcomp^i$ as was done in
in~\cite{CohenCoquandHuberMortberg15}.)

Given $\Gamma, x : \sone \der C$, $\Gamma \der b : C \subst x \base$,
and $\Gamma \der l : \Path^i\, C \subst x {\sloop i}\, b\,b$ we also
add the reduction rules for the elimination
\begin{align*}
  \Gamma &\der \sonerec_{x.C}\,\base\,b\,l \red b : C \subst x \base
  \\
  \Gamma &\der \sonerec_{x.C}\,(\sloop r)\,b\,l \red l \, r : C \subst
  x {\sloop r}
\end{align*}
where $\Gamma \der r \neq 1 : \II$, and moreover for $\Gamma \der \phi
\neq 1 : \FF$,
\begin{multline*}
  \Gamma \der \sonerec_{x.C}\,(\comp^i\,\sone\,[\phi \mapsto
  u]\,u_0)\,b\,l \red \\ \comp^i\,{C \subst x v}\,[\phi \mapsto
  u']\,u_0': C \subst x {\comp^i\,\sone\,[\phi \mapsto u]\,u_0}
\end{multline*}
where $v = \Comp^i\,\sone\,[\phi \mapsto u]\,u_0$, $u' =
\sonerec_{x.C}\,u\,b\,l$, $u_0' = \sonerec_{x.C}\,u_0\,b\,l$, and we
assumed $i \notin \dom \Gamma$ (otherwise rename $i$).

Furthermore, if $\Gamma \der t \red t' : \sone$, then
\begin{equation*}
  \Gamma \der \sonerec_{x.C}\,t\,b\,l \red \sonerec_{x.C}\,t'\,b\,l : C
  \subst x {t'}.
\end{equation*}

Consequently, we also call expressions introduced if they are of the
form $\sone$, $\base$, $\sloop r$ with $r \notin \set {0,1}$, and
$\comp^i\,\sone\,[\phi \mapsto u]\,u_0$ with $\phi \neq 1$.

Next, the computability predicates and relations are adapted as
follows: $I \lforce \sone$ and $I \lforce \sone = \sone$.  $I \lforce
u : \sone$ and $I \force u = v : \sone$ are defined simultaneously
(similarly as for $\N$):
\begin{mathparpagebreakable}
  \inferrule { } {I \lforce \base : \sone} %
  \and %
  \inferrule { r \in \II (I) - \set{0,1} \\
    I \lforce \sloop 0 : \sone\\
    I \lforce \sloop 1 : \sone } {I \lforce \sloop r : \sone} %
  \and %
  \inferrule {%
    1\neq \phi \in \FF (I) \\
    I,i,\phi \lforce u : \sone\\
    I \lforce u_0 : \sone\\
    I,\phi \lforce u_0 = u(i0) : \sone\\
    I,\phi \lforce \comp^i\,\sone\,[\phi \mapsto u]\,u_0 : \sone } %
  { I \lforce \comp^i\,\sone\,[\phi \mapsto u]\,u_0 : \sone} %
  \and %
  \inferrule {u \noti \\
    \all f \co J \to I (u f \hr^\sone \And J \lforce u f \dn^\sone :
    \sone) \\
    \all f\co J \to I \all g \co K \to J (K \lforce u f \dn g = u fg
    \dn : \sone) } %
  { I \lforce u : \sone}
\end{mathparpagebreakable}
Note, the (admissible) two last premises in the case for $\sloop$ are
there to not increase the height of the derivation when doing a
substitution (Lemma~\ref{lem:subst}); similarly for the last premise
in the rule for composition.  The relation $I \lforce u = v : \sone$
is defined analogously, that is, by the usual congruence rules and a
clause for when $u$ or $v$ is non-introduced as we have it for~$\N$
(see also the next section).  To adapt
Theorem~\ref{thm:comp-is-computable} note that compositions are
computable for $\phi = 1$ by using the Expansion Lemma and the
reduction rule; using this, compositions are computable by definition
also for $\phi \neq 1$.

\subsection{Propositional Truncation}
\label{sec:inh}

We will use a slight simplification of propositional truncation as
presented in~\cite[Section~9.2]{CohenCoquandHuberMortberg15}.  Let us
thus recall the typing rules (omitting congruence rules): the
formation rule is $\Gamma \der \inh A$ whenever $\Gamma \der A$, and
likewise $\Gamma \der \inh A : \U$ whenever $\Gamma \der A :
\U$.  Moreover:
\begin{mathpar}
  \inferrule %
  { \Gamma \der a : A} %
  { \Gamma \der \inc a : \inh A} %
  \and
  \inferrule %
  { \Gamma \der u : \inh A \\
    \Gamma \der v : \inh A \\
    \Gamma \der r : \II
  } %
  { \Gamma \der \squash\,u\,v\,r : \inh A
  } %
  \and %
  \and %
  \inferrule %
  { \Gamma \der A\\
    \Gamma \der \phi : \FF\\
    \Gamma,\phi, i:\II \der u : \inh A\\
    \Gamma \der u_0 : \inh A [\phi \mapsto u \subst i 0]
  } %
  { \Gamma \der \hcomp^i_{\inh A}\,[\phi \mapsto u]\,u_0 : \inh A
  }
\end{mathpar}
with the judgmental equalities (omitting context and type):
\begin{mathpar}
  \squash\,u\,v\,0 = u %
  \and %
  \squash\,u\,v\,1 = v %
  \and %
  \hcomp^i_{\inh A}\,[1_\FF \mapsto u]\,u_0 = u \subst i 1
\end{mathpar}
Note that the type in $\hcomp^i$ does not depend on $i$ and we call
these \emph{homogeneous compositions}.  The eliminator, given $\Gamma
\der A$ and $\Gamma, z : \inh A \der C (z)$, is given by the rule
\begin{mathpar}
  \inferrule %
  { 
    \Gamma \der w : \inh A\\
    \Gamma \der t : (a : A) \to C (\inc a)\\
    \Gamma \der p : (u\, v : \inh A) (x : C (u)) (y : C (v)) \to
    \Path^i\,(C (\squash\,u\,v\,i))\,x\,y\\
  } %
  { \Gamma \der \ptElim_{z.C}\,w\,t\,p : C (w)} %
\end{mathpar}
together with judgmental equalities (assuming $i$ fresh):
\begin{align*}
  \ptElim_{z.C}\,(\inc a)\,t\,p &= t\,a \\
  \ptElim_{z.C}\,(\squash\,u\,v\,r)\,t\,p &=
  p\,u\,v\,(\ptElim_{z.C}\,u\,t\,p)\,(\ptElim_{z.C}\,v\,t\,p)\,r\\
  \ptElim_{z.C}\,(\hcomp^i\,[\phi\mapsto u]\,u_0)\,t\,p &=
  \comp^i\,{C \subst z w}\,[\phi\mapsto
  \ptElim_{z.C}\,u\,t\,p]\,(\ptElim_{z.C}\,u_0\,t\,p)
\end{align*}
where $w = \hcomp^j\,[\phi\mapsto u \subst i {i \land j}, (i=0)
\mapsto u_0]\,u_0$.

Instead of $\mathsf{transp}$ and $\mathsf{squeeze}$
in~\cite{CohenCoquandHuberMortberg15} we take the following
\emph{forward} operation:
\begin{mathpar}
  \inferrule %
  { \Gamma, i : \II \der A  \\
    \Gamma \der r : \II \\
    \Gamma \der u : \inh {A \subst i r}
  } %
  { \Gamma \der \fwd_{i.A}\,r\,u : \inh {A \subst i 1}
  } %
\end{mathpar}
which comes with the judgmental equalities:
\begin{align}{}
  \nonumber
  \fwd\,1\,u
  &= u
  \\
  \label{eq:fwd-inc}
  \fwd\,r\,(\inc a)
  &= \inc (\comp^i\,{A \subst i {i \lor r}}\,[(r=1) \mapsto a]\,a)
  \\
  \label{eq:fwd-squah}
  \fwd\,r\,(\squash\,u\,v\,s)
  &= \squash\,(\fwd\,r\,u)\,(\fwd\,r\,v)\,s
  \\
  \label{eq:fwd-hcomp}
  \fwd\,r\,(\hcomp^j_{\inh {A \subst i r}}\,[\phi \mapsto
    u]\,u_0)
  &= \hcomp^j_{\inh {A \subst i 1}}\,[\phi \mapsto
  \fwd\,r\,u]\,(\fwd\,r\,u_0)
\end{align}

Composition for $\inh A$ is now explained using $\fwd$ and homogeneous
composition:
\begin{equation*}
  \comp^i\,{\inh A}\,[\phi \mapsto u]\,u_0 = \hcomp^i_{\inh {A \subst
      i 1}}\,[\phi \mapsto \fwd_{j.{A\subst i
      j}}\,i\,u]\,(\fwd_{i.A}\,0\,u_0)
\end{equation*}

Next, we extend the reduction relation by directing the above
judgmental equalities from left to right, but requiring the following
extra conditions to guarantee determinism (additionally to the
suppressed well-typedness). The directed versions
of~\eqref{eq:fwd-inc}--\eqref{eq:fwd-hcomp} require $r \neq 1$;
\eqref{eq:fwd-squah} and \eqref{eq:fwd-hcomp} additionally require $s
\neq 1$ and $\phi \neq 1$, respectively.  Similarly for the reductions
of $\ptElim$.  Additionally, we need congruence rules:
\begin{mathpar}
  \inferrule %
  { \Gamma, i : \II \der A \\
    \Gamma \der r \neq 1 : \FF \\
    \Gamma \der u \red v : \inh {A \subst i r}} %
  { \Gamma \der \fwd_{i.A}\,r\,u \red \fwd_{i.A}\,r\,v : \inh {A \subst
  i 1}}
\end{mathpar}
and a similar such rule for $\ptElim$.  Correspondingly, we also call
expressions of the following form introduced: $\inh A$, $\inc\,a$,
$\squash\,u\,v\,r$ with $r \neq 1$, and $\hcomp$'s with $\phi \neq 1$.

To incorporate propositional truncation in the computability
predicates we add new the formation rules:
\begin{mathpar}
  \nir {\rcpt} {I,1 \lforce A} {I \lforce \inh A} %
  \and%
  \nir {\rept} {I \lforce A = B} {I \lforce \inh A = \inh B} %
\end{mathpar}
And in the case $I \lforce A$ was derived via $\rcpt$ the definition
of $I \lforce u : A$ is extended to:
\begin{mathpar}
  \inferrule %
  { I \lforce a : A} %
  { I \lforce \inc a : \inh A} %
  \and%
  \inferrule %
  {  0,1 \neq r \in \II (I) \\
    I \lforce \squash\,u\,v\,0 : \inh A \\
    I \lforce \squash\,u\,v\,1 : \inh A } %
  { I \lforce \squash\,u\,v\,r : \inh A} %
  \and%
  \inferrule {%
    1\neq \phi \in \FF (I) \\
    I,i,\phi \lforce u : {\inh A}\\
    I \lforce u_0 : {\inh A}\\
    I,\phi \lforce u_0 = u(i0) : {\inh A}\\
    I,\phi \lforce \hcomp^i_{\inh A}\,[\phi \mapsto u]\,u_0 : {\inh A}
  } %
  { I \lforce \hcomp^i_{\inh A}\,[\phi \mapsto u]\,u_0 : {\inh A} } %
  \and %
  \inferrule {u \noti \\
    \all f \co J \to I (u f \hr^{\inh {Af}} \And J \lforce u f
    \dn^{\inh {A f}} :
    {\inh {A f}}) \\
    \all f\co J \to I \all g \co K \to J (K \lforce u f \dn g = u fg
    \dn : {\inh {A fg}}) } %
  { I \lforce u : {\inh A}}
\end{mathpar}
As before, the rather unnatural formulation of the rules for $\squash$
and $\hcomp$ is to ensure that the height of a derivation is not
increased after performing a substitution (Lemma~\ref{lem:subst}).
\begin{mathpar}
  \inferrule %
  { I \lforce a = a' : A} %
  { I \lforce \inc a = \inc {a'} : \inh A} %
  \and%
  \inferrule %
  {  0,1 \neq r \in \II (I) \\
    I \lforce u = u' : \inh A \\
    I \lforce v = v' : \inh A } %
  { I \lforce \squash\,u\,v\,r = \squash\,u'\,v'\,r : \inh A} %
  \and%
  \inferrule {%
    1\neq \phi \in \FF (I) \\
    I,i,\phi \lforce u = u' : {\inh A}\\
    I \lforce u_0 = u_0': {\inh A}\\
    I,\phi \lforce \hcomp^i_{\inh A}\,[\phi \mapsto u]\,u_0 =
    \hcomp^i_{\inh A}\,[\phi \mapsto u']\,u_0' : {\inh A}
  } %
  { I \lforce \hcomp^i_{\inh A}\,[\phi \mapsto u]\,u_0 =
    \hcomp^i_{\inh A}\,[\phi \mapsto u']\,u_0' : {\inh A} } %
  \and %
  \inferrule {u \text{ or }u'\noti \\
    \all f \co J \to I (J \lforce u f
    \dn^{\inh {A f}} = u' f
    \dn^{\inh {A f}} :
    {\inh {A f}})} %
  { I \lforce u = u': {\inh A}}
\end{mathpar}

We now sketch how one can extend the proofs of
Sections~\ref{sec:computability} and~\ref{sec:soundness}.  The
additional case in the Expansion Lemma is handled as for natural
numbers.  Next, one proves the introduction rules for $\inc$,
$\squash$, and $\hcomp$ correct.  To handle the new case $\rcpt$ for
propositional truncation in Theorem~\ref{thm:comp-is-computable} one
has to simultaneously prove
\begin{align*}
  I \force u : \inh {A \subst i r}
  &\Imp I \force \fwd\,r\,u : \inh {A \subst i 1}\\
  I \force u = v : \inh {A \subst i r}
  &\Imp I \force \fwd\,r\,u = \fwd\,r\,v : \inh {A \subst i 1}
\end{align*}
by a side induction on the premises. Finally, one can then show
soundness of $\ptElim$.

We not only get the corresponding canonicity result for the extended
theory, but we can also extract witnesses from $\inh A$ as long as we
are in a name context:
\begin{theorem}
  \label{thm:inh-witnesses}
  If $I \der A$ and $I \der u : \inh A$, then $I \der v : A$ for some
  $v$, where $I$ is a context of the form $i_1:\II,\dots, i_n : \II$
  with $n \ge 0$.
\end{theorem}
\begin{proof}
  By Soundness we get $I \models A$ and $I \models u : \inh A$, and
  hence also $I \force A$ and $I \force u : \inh A$.  By induction on
  $I \force u : \inh A$ we show that there is some $v$ such that $I
  \force v : A$.  In the case for $\inc$ this is direct; any other
  case follows from the \IH.  Thus also $I \der v : A$ as required.
\end{proof}

As a direct consequence we get that the logic of mere propositions
(cf.\ \cite[Section~3.7]{HottBook13}) of cubical type theory satisfies
the following \emph{existence property}. Define $\exists (x : A)\, B$
as the truncated $\Sigma$-type, i.e., $\inh{(x : A) \times B}$.

\begin{corollary}
  If $I \der \exists (x : A)\, B(x)$ is true (i.e., there is a term
  inhabiting the type), then there exists $u$ with $I \der u : A$ such
  that $I \der B \subst x u$ is true, where $I$ is a context of the
  form $i_1:\II,\dots, i_n : \II$ with $n \ge 0$.
\end{corollary}

\section{Conclusion}
\label{sec:conclusion}

We have shown canonicity for cubical type
theory~\cite{CohenCoquandHuberMortberg15} and its extensions with the
circle and propositional truncation.  This establishes that the
judgmental equalities of the theory are sufficient to compute closed
naturals to numerals; indeed, we have even given a deterministic
reduction relation to do so.  It should be noted that we could have
also worked with the corresponding \emph{untyped} reduction relation
$A \red B$ and then take $I \der A \red B$ to mean $I \der A = B$, $I
\der A$, $I \der B$, and $A \red B$ etc.

To prove canonicity we devised computability predicates (and
relations) which, from a set-theoretic perspective, are constructed
using the \emph{least} fixpoint of a suitable operator.  It is
unlikely that this result is optimal in terms of proof-theoretic
strength; we conjecture that it is possible to modify the argument to
only require the existence of a fixpoint of a suitably modified
operator (and not necessarily its least fixpoint); this should be
related to how canonicity is established in~\cite{AngiuliHarper16}.

We expect that the present work can be extended to get a normalization
theorem and to establish decidability of type checking for cubical
type theory (and proving its implementation\footnote{Available at
  \url{https://github.com/mortberg/cubicaltt}.}  correct).  One new
aspect of such an adaption is to generalize the computability
predicates and relations to expressions in any contexts in which we
get new introduced expressions given by systems; moreover, we will
have to consider reductions in such general contexts as well which has
to ensure that, say, variables of path-types compute to the right
endpoints.

Another direction of future research is to investigate canonicity of
various extensions of cubical type theory, especially adding resizing
rules.

\begin{acknowledgments*}
  I thank Carlo Angiuli, Thierry Coquand, Robert Harper, and Bassel
  Mannaa for discussions about this work, as well as Milly Maietti who
  also suggested to investigate the existence property.  I am also
  grateful for the comments by the anonymous reviewer.
\end{acknowledgments*}


\bibliographystyle{amsplain}
\bibliography{cttnormal}

\end{document}